\documentclass[11pt]{article}

\usepackage[margin=1in]{geometry}

\usepackage{amsmath,amssymb,amsthm,mathtools}
\usepackage{enumitem}
\usepackage{booktabs}

\usepackage{algorithm}
\usepackage{algpseudocode}
\floatstyle{ruled}
\restylefloat{algorithm}

\usepackage{tikz}
\usetikzlibrary{arrows.meta,positioning}

\usepackage[hidelinks]{hyperref}
\usepackage{csquotes}
\usepackage[numbers,sort&compress]{natbib}
\newcommand{\textcite}[1]{\citet{#1}}
\newcommand{\parencite}[1]{\citep{#1}}

\usepackage{microtype}

\newcommand{\instsize}[1]{\lVert #1\rVert}

\theoremstyle{definition}
\newtheorem{definition}{Definition}[section]
\newtheorem{problem}{Problem}[section]
\newtheorem{remark}{Remark}[section]
\newtheorem{example}{Example}[section]

\theoremstyle{plain}
\newtheorem{theorem}[definition]{Theorem}
\newtheorem{lemma}[definition]{Lemma}
\newtheorem{proposition}[definition]{Proposition}
\newtheorem{corollary}[definition]{Corollary}

\newcommand{\A}{\mathsf{A}}

\newcommand{\Vars}{\mathsf{V}}
\newcommand{\Fun}{\mathsf{F}}
\newcommand{\Sol}{\mathrm{Sol}}
\newcommand{\Img}{\mathrm{Im}}
\newcommand{\Disp}{\mathrm{Disp}}
\newcommand{\PerfSpec}{\operatorname{Spec}}
\newcommand{\Sn}{S_n}
\DeclareMathOperator{\Guess}{Guess}
\DeclareMathOperator{\Def}{Def}
\DeclareMathOperator{\Src}{Src}

\title{\textbf{Term Coding and Dispersion:}\\
\textbf{Exact and Asymptotic Decision Problems}}

\author{S\o ren Riis\\
Queen Mary University of London\\
\texttt{s.riis@qmul.ac.uk}}

\date{}

\begin{document}
\maketitle

\begin{abstract}
Let $\mathbf t=(t_1,\ldots,t_r)$ be a tuple of terms defining, under an interpretation on an
$n$-element alphabet, a map $A^k\to A^r$.  We study the decision theory of its maximum image size
$\Disp_n(\mathbf t)$, separating exact perfect dispersion from asymptotic rate.

Building on the term-cut theorem of Riis and Gadouleau, we prove that every eventual threshold strictly
between consecutive integer powers is decidable in polynomial time: if
$n^d<h_d(n)=o(n^{d+1})$, then $\Disp_n(\mathbf t)\ge h_d(n)$ eventually if and only if the term-cut
exponent satisfies $D(\mathbf t)\ge d+1$.  For the exact problem, we introduce the perfect-alphabet spectrum
and prove that it is multiplicatively closed, that a nonempty spectrum forces full rate, and that the converse
fails.  We completely characterize the one-output case.

On square instances, perfect dispersion is precisely finite square term bijectivity.  We give explicit
linear-size padding reductions from three-dimensional square bijectivity to perfect dispersion with every fixed
output dimension $r\ge3$.  We also characterize scalar-linear witnesses by a determinant polynomial, obtaining
decidability over fixed finite fields, over extensions of a fixed characteristic, and over arbitrary finite
fields.  General square bijectivity remains open.  The principal mathematical results have been machine-checked
in Lean.
\end{abstract}


\section{Introduction}
\label{sec:intro}
Many problems in coding theory, network coding, and extremal combinatorics can be phrased as a
\emph{code design} problem:
we are given a finite set of local functional constraints, and we must choose the local functions
so as to maximise the number of global configurations consistent with all constraints.
Term coding provides a uniform equational language for this optimisation problem.  The present paper studies
its dispersion subclass, in which a tuple of terms defines a finite map and the objective is its maximum image
size.  The main point is a sharp separation between two questions that initially look similar: asymptotic growth
is governed by a polynomial-time min-cut invariant, whereas exact surjectivity leads to multiplicative spectra,
finite square bijectivity, and a substantial open decision problem.  Our contributions make this separation
precise through threshold, spectrum, transfer, and scalar-linear classification theorems.

\subsection{The computational problem}
Fix a finite alphabet $\A$ of size $n$.
(When convenient---for example, when we partition $\A$ into blocks---we identify $\A$ with
$[n]:=\{0,1,\dots,n-1\}$.)
A \emph{term-coding instance} consists of a finite set of variables
$\Vars=\{x_1,\dots,x_v\}$, a finite signature $\Fun$ of function symbols with specified arities,
and a finite set of constraints given as term equations
\[
\Gamma=\{\, s_1=t_1,\dots,s_m=t_m\,\},
\]
where each $s_i,t_i$ is a term built from $\Vars$ and symbols in $\Fun$.
An \emph{interpretation} $\mathcal I$ assigns to each $k$-ary symbol $f\in\Fun$ a concrete total function
$f^\mathcal I:\A^k\to \A$.

The variables and function symbols play different mathematical roles.  The variables in $\Vars$ range over
$\A$ and determine the assignments counted as codewords, whereas the optimiser chooses the interpretations of
the symbols in $\Fun$.  Thus the local encoding and decoding functions are part of a feasible design, as in
network coding.

Given an assignment $\mathbf a\in\A^{\Vars}$, every term $u$ over $(\Vars,\Fun)$ evaluates to an element
$u^{\mathcal I}(\mathbf a)\in\A$ by the usual recursion on the syntax of $u$ (variables are read from $\mathbf a$
and function symbols are applied using $\mathcal I$). We write $\mathcal I \models s[\mathbf a]=t[\mathbf a]$
to mean $s^{\mathcal I}(\mathbf a)=t^{\mathcal I}(\mathbf a)$.

For a fixed $\mathcal I$ we write
\[
\Sol_\mathcal I(\Gamma;n)
:=\bigl\{\mathbf a\in \A^{\Vars} : \mathcal I \models s_i[\mathbf a]=t_i[\mathbf a]\text{ for all }i\bigr\}
\]
for the set of satisfying assignments (codewords).
The central optimisation quantity is
\[
\Sn(\Gamma)\ :=\ \max_{\mathcal I}\ |\Sol_\mathcal I(\Gamma;n)|.
\]
\begin{problem}[TERM-CODING-MAX]
\label{prob:termcoding-max}
\textbf{Input:} a finite term-coding instance $\Gamma$ and an alphabet size $n\ge 2$.
\textbf{Output:} the maximum code size $S_n(\Gamma)$.
\end{problem}
\begin{example}[A toy index-coding instance]
\label{ex:indexcoding}
Let $\Vars=\{x_1,x_2,x_3,y\}$ and let $\Fun$ contain a ternary symbol $f$ (an \emph{encoder}) and three binary
symbols $h_1,h_2,h_3$ (\emph{decoders}).  Consider the equations
\[
  y=f(x_1,x_2,x_3),\qquad
  x_1=h_1(y,x_2),\qquad
  x_2=h_2(y,x_3),\qquad
  x_3=h_3(y,x_1).
\]
An interpretation $\mathcal I$ chooses the broadcast rule $f^{\mathcal I}:\A^3\to\A$ and the local decoding
rules $h_i^{\mathcal I}:\A^2\to\A$.
For a fixed alphabet size $n$, the optimised value $S_n(\Gamma)$ asks: \emph{how many message triples}
$(x_1,x_2,x_3)\in\A^3$ can be supported by a single broadcast symbol $y\in\A$ when receiver~$i$ knows
$x_{i+1}$ as side information (indices modulo $3$)?
This is a small term-coding instance of the same flavour as index coding: one chooses an
encoding and decoding scheme to maximise the set of simultaneously valid global configurations.
\end{example}

\begin{remark}[Combinatorial graded feasibility (omitted here)]
Term coding also yields graded versions of classical existence problems in algebraic combinatorics
(e.g.\ Steiner quasigroup identities, Latin squares, and related design axioms), by maximising the fraction of
assignments on which a set of identities holds.
We do not pursue those examples here, since they are not needed for the dispersion and complexity results below.
See \parencite{riis2026termcoding} for related graded-feasibility examples and a broader entropic viewpoint.
\end{remark}

\subsection{A structural invariant: dependency graphs and guessing numbers}
A key step is to convert $\Gamma$ into a normal form that exposes its functional dependencies
(replacing nested terms by auxiliary variables), and then into a \emph{diversified} form where each
function occurrence is given its own symbol.
From the diversified normal form we extract a directed \emph{dependency graph}
$G_\Gamma$ whose vertices are variables and whose edges encode the relation ``this variable is functionally determined by
those variables''.
This makes $\Sn(\Gamma)$ comparable to the number of winning configurations in a \emph{guessing game} on
$G_\Gamma$, and hence to the graph's \emph{guessing number}.
Guessing numbers are central in network and index coding
\parencite{riis2006information,riis2007graph,gadouleau2011graph} and naturally quantify the maximum joint entropy
consistent with deterministic dependencies.

\subsection{The complexity landscape for dispersion}
We then focus on \emph{dispersion}, where the constraints describe the image size of a term-defined map.
A (single-sorted) dispersion instance is given by $k$ input variables
$\mathbf x=(x_1,\dots,x_k)$ and a list of $r$ output terms
$\mathbf t(\mathbf x)=(t_1(\mathbf x),\dots,t_r(\mathbf x))$, inducing under interpretation $\mathcal I$ the map
\[
\Theta^\mathcal I:\A^k\to \A^r,\qquad
\mathbf a\mapsto \bigl(t_1^\mathcal I(\mathbf a),\dots,t_r^\mathcal I(\mathbf a)\bigr).
\]
The dispersion objective is
\[
\Disp_n(\mathbf t)\ :=\ \max_{\mathcal I}\ |\Img(\Theta^\mathcal I)|.
\]

We also consider the associated \emph{dispersion exponent}
\[
D(\mathbf t)\ :=\ \lim_{n\to\infty}\frac{\log \Disp_n(\mathbf t)}{\log n},
\]
which exists for every dispersion instance and is an integer equal to a minimum term-cut
(Theorem~\ref{thm:disp-exponent}).

\begin{example}[A single-function \emph{diamond} gadget]
\label{ex:diamond-intro}
Riis and Gadouleau~\parencite{riis2019max} exhibited a related four-term, single-function instance with the
same full-rate-but-not-perfect phenomenon.  The following variant admits a particularly short collision proof.
Let $x,y,z,w\in\A$ be inputs and let $f$ be a \emph{single} binary function symbol.
Consider the $4$-tuple of output terms
\[
\mathbf t(x,y,z,w)\ :=\ \bigl( f(x,y),\, f(x,z),\, f(y,w),\, f(z,w)\bigr),
\]
which induces under an interpretation $\mathcal I$ a map $\Theta^{\mathcal I}:\A^4\to\A^4$.
For $n\ge2$, this map can \emph{never} be bijective (hence never surjective): restrict to inputs with $y=z$.
There are $n^3$ such inputs $(x,y,z,w)=(x,y,y,w)$, but the corresponding outputs have the repeated form
$(a,a,b,b)$ and therefore take at most $n^2$ distinct values.
Collisions are thus unavoidable by the pigeonhole principle.
On the other hand, the term-cut exponent of this instance is full, $D(\mathbf t)=4$
(see Theorem~\ref{thm:disp-exponent} and Lemma~\ref{lem:diamond-maxflow}).
So at the level of coarse asymptotic rates the gadget behaves as though it has full dimension, even though perfect
dispersion is impossible.
\end{example}

Dispersion exhibits a sharp contrast between the computability of its asymptotic and exact questions:
\begin{itemize}[leftmargin=*]
\item \textbf{Every gap between consecutive integer powers has a polynomial-time threshold test.}
The term-cut theorem makes $D(\mathbf t)$ an integer computable by max-flow.  We turn this structural fact into
an exact decision theorem: for every $h_d$ with $n^d<h_d(n)=o(n^{d+1})$, the eventual inequality
$\Disp_n(\mathbf t)\ge h_d(n)$ holds exactly when $D(\mathbf t)\ge d+1$
(Theorem~\ref{thm:ptime-threshold}).

\item \textbf{Perfect dispersion has an arithmetic spectrum.}
The set of alphabet sizes on which a tuple is perfectly dispersive is closed under multiplication.  If this
spectrum is nonempty, then the tuple has full asymptotic rate; the converse fails.  For one output we determine
the spectrum completely (Theorem~\ref{thm:perfect-spectrum-structure} and
Theorem~\ref{thm:one-output}).

\item \textbf{Square perfect dispersion is square bijectivity.}
On instances with $k$ inputs and $k$ outputs, perfect dispersion is equivalent to the existence of a finite
interpretation making the induced map $A^k\to A^k$ bijective.  Explicit linear-size padding reductions transfer
$\textsc{BIJ}_3$ to every fixed output dimension $r\ge3$ (Theorem~\ref{thm:reduction-nr}).

\item \textbf{The scalar-linear fragment has an exact algebraic criterion.}
For square tuples, scalar-linear bijectivity is equivalent to the nonvanishing of a determinant polynomial.
This gives decision procedures over each fixed finite field, over extensions of a fixed characteristic, and
over some finite field of arbitrary characteristic (Theorem~\ref{prop:scalar-det-criterion-main} and
Corollary~\ref{cor:scalar-linear-classification}).
\end{itemize}

\subsection{Main results}
\label{subsec:new-results}
The main contributions fall into five groups.  Together they distinguish the asymptotic problem, controlled by
a term-cut, from the exact problem, controlled by finite-alphabet structure.  Sections~\ref{sec:tc}--\ref{sec:dispersion}
develop the foundations used in their proofs; the relation of those foundations to earlier work is stated in
Section~\ref{subsec:provenance}.
\begin{enumerate}[label=(\textbf{N\arabic*}), leftmargin=*]
\item \emph{Threshold classification.}
      Theorem~\ref{thm:ptime-threshold} classifies every eventual threshold strictly between $n^d$ and
      $n^{d+1}$ and yields a polynomial-time decision algorithm.  The exact constant-one term-cut upper bound is
      essential at the lower edge $n^d+1$.
\item \emph{Perfect-alphabet spectra.}
      Definition~\ref{def:perfect-spectrum} packages exact surjectivity across all finite alphabet sizes.
      Theorem~\ref{thm:perfect-spectrum-structure} proves multiplicative closure and the implication from a
      nonempty spectrum to full rate, while the diamond proves that full rate is insufficient.  Theorem~\ref{thm:one-output}
      gives the complete one-output classification.
\item \emph{Square-bijectivity transfer.}
      Lemma~\ref{lem:square-perfect-equals-bij} identifies the square exact problem, and
      Theorems~\ref{thm:reduction-kk} and~\ref{thm:reduction-nr} give explicit linear-size padding reductions
      from three-dimensional square bijectivity to every fixed output dimension $r\ge3$.  These are unconditional
      equivalence and reduction theorems; undecidability is not assumed.
\item \emph{Scalar-linear classification.}
      Theorem~\ref{prop:scalar-det-criterion-main} associates a determinant polynomial to every square tuple.
      Corollary~\ref{cor:scalar-linear-classification} then resolves scalar-linear existence over fixed finite
      fields, fixed-characteristic extensions, and finite fields of arbitrary characteristic.
\item \emph{Explicit least-alphabet results.}
      Proposition~\ref{prop:small-spectra} gives two two-variable tuples whose least perfect alphabet sizes are
      respectively three and four.  These examples show that the exact spectrum contains information invisible
      to the exponent alone.
\end{enumerate}

\subsection{Relation to earlier work and self-contained scope}
\label{subsec:provenance}
Two earlier sources provide the foundations used here.  Riis and Gadouleau~\parencite{riis2019max} introduced
dispersion for term-defined communication networks and proved the integral term-cut exponent theorem and its
max-flow computation.  They also exhibited a related four-term, single-function instance with full asymptotic
rate but no perfect interpretation.  The present paper is a self-contained continuation of the Term Coding
framework developed in the earlier paper~\parencite{riis2026termcoding}.  That paper establishes the general
equational framework, the normalisation and diversification pipeline, the guessing-game sandwich, the existence
of term-coding rates, and the embedding of dispersion into that framework.

For self-containment, Sections~\ref{sec:tc}--\ref{sec:dispersion} recall these ingredients, give complete proofs
of the forms used here, and make the max-flow construction explicit.  Relative to the 2019 paper, the genuinely
new material is the threshold decision theorem, the exact spectrum and one-output results, the square-padding
transfer, and the scalar-linear and finite-spectrum analysis.  Relative to the earlier Term Coding paper, the
genuinely new material is the exact-versus-asymptotic decision theory developed in
Sections~\ref{sec:complexity-dichotomy} and~\ref{sec:bij-status}.  This division of provenance is maintained
throughout the statements and citations.

The paper makes no converse reduction from arbitrary rectangular perfect dispersion to square bijectivity and
no unconditional undecidability claim.  The decidability of unrestricted square bijectivity and of
$\textsc{BIJ}_k$ for fixed $k\ge2$ remains open.

\subsection{Machine verification}
\label{subsec:verification}
The principal mathematical results of this paper have been formalised and machine-checked in Lean~4 against
Mathlib.  The development builds on the machine-checked formalisation for the earlier Term Coding paper
\parencite{riis2026termcoding}, available at
\url{https://github.com/SR123/term-coding-lean}
(DOI \texttt{10.5281/zenodo.22132608}).  The development for the present paper extends it and is archived at:
\begin{quote}
\noindent\textbf{GitHub:} \url{https://github.com/SR123/term-coding-dispersion-lean}\\
\textbf{DOI:} \href{https://doi.org/10.5281/zenodo.22228344}{\texttt{10.5281/zenodo.22228344}}
\end{quote}
The term-cut upper bound and the tagged-alphabet routing lower bound used in
Theorem~\ref{thm:disp-exponent} have also been checked in this development, so the theorem is machine-checked
modulo Menger's theorem.  The product-alphabet construction underlying
Theorem~\ref{thm:perfect-spectrum-structure}\textit{(i)} has likewise been checked.  The polynomial-time cost
analyses, the matrix-linear perspective in Section~\ref{subsec:matrix-linear}, and the literature-status statements in
Section~\ref{sec:bij-status} are outside the machine-checked scope.

\subsection{Roadmap}
Section~\ref{sec:tc} defines term coding formally.
Section~\ref{sec:preprocess} describes the preprocessing pipeline and the dependency graph.
Section~\ref{sec:guess} recalls guessing games and states the guessing-number sandwich bound.
Section~\ref{sec:dispersion} introduces dispersion, proves the term-cut exponent theorem, and shows how
dispersion embeds into term coding.
Section~\ref{sec:complexity-dichotomy} proves the threshold and spectrum theorems, identifies square perfect
dispersion with square bijectivity, and gives the padding reductions.
Section~\ref{sec:bij-status} proves the low-dimensional and scalar-linear classifications and then delineates
the unresolved square-bijectivity frontier.

\subsection{Notation quick reference}
\label{sec:notation}
\begin{center}
\begin{tabular}{@{}ll@{}}
\toprule
Symbol & Meaning \\
\midrule
$\A$ & finite alphabet, $|\A|=n$ \\
$\Vars$ & set of term-coding variables, $|\Vars|=v$ \\
$\Fun$ & signature of function symbols (symbols + arities) \\
$\Gamma$ & term-coding instance (set of equations over $\Vars$ and $\Fun$) \\
$\Sol_{\mathcal I}(\Gamma;n)$ & solutions over alphabet size $n$ under interpretation $\mathcal I$ \\
$S_n(\Gamma)=\Sn(\Gamma)$ & maximum code size (optimised over interpretations) \\
$\Gamma^{\mathrm{cf}}$ & collision-free normal form instance \\
$\Gamma^{\mathrm{div}}$ & diversified instance \\
$G_\Gamma$ & dependency digraph of a functional instance \\
$\Src(\Gamma)$ & sources (variables not defined by any equation) \\
$W(G,S;n)$ & winning configurations in guessing game on $(G,S)$ \\
$\Guess(G,S;n)$ & guessing number $\log_n W(G,S;n)$ \\
$\mathbf t(\mathbf x)$ & dispersion instance (output terms in input variables $\mathbf x$) \\
$\Disp_n(\mathbf t)$ & maximum image size (dispersion) at alphabet size $n$ \\
$\PerfSpec(\mathbf t)$ & perfect-alphabet spectrum $\{n\ge2:\Disp_n(\mathbf t)=n^r\}$ \\
$D(\mathbf t)$ & dispersion exponent (integer; minimum term-cut, computable by max-flow) \\
\bottomrule
\end{tabular}
\end{center}

\section{Term coding instances}
\label{sec:tc}

We work in a finitary single-sorted signature in the sense of universal algebra/first-order logic.
For consistency with standard universal-algebra conventions, we adopt the following:
the signature includes $0$-ary symbols (constants) and is part of the instance description.

\subsection{Terms and instances}
A \emph{signature} $\Fun$ consists of function symbols, each with a fixed arity (including arity $0$).
The set of \emph{terms} over $\Fun$ and variables $\Vars$ is defined inductively:
variables are terms, and if $f\in\Fun$ is $k$-ary and $u_1,\dots,u_k$ are terms then
$f(u_1,\dots,u_k)$ is a term.

\begin{definition}[Term-coding instance]
A \emph{term-coding instance} consists of:
\begin{itemize}[leftmargin=*]
\item a finite set of variables $\Vars$;
\item a finite signature $\Fun$ (symbols together with their arities);
\item a finite set of term equations $\Gamma=\{s_i=t_i: i\in[m]\}$ over $\Vars$ and $\Fun$.
\end{itemize}
When $\Vars$ and $\Fun$ are clear from context we refer to the instance simply by $\Gamma$.
\end{definition}

\begin{definition}[Syntactic size of an instance]
\label{def:instance-size}
Let $\Gamma=\{\,s_i=t_i : i\in[m]\,\}$ be a term-coding instance.
We define the \emph{instance size} $\instsize{\Gamma}$ to be the total number of symbol occurrences in
all terms $s_i,t_i$ (counting each variable and each function-symbol occurrence once) plus the number of
equations $m$.
Equivalently, if each term is represented as a rooted syntax tree and we take a disjoint union of all
those trees, then $\instsize{\Gamma}$ is the number of nodes in that forest plus $m$.

For a dispersion instance $\mathbf t=(t_1,\dots,t_r)$ presented as a shared term DAG, we define
$\instsize{\mathbf t}$ to be the number of DAG nodes plus $r$; a syntax tree is the special case with no
sharing.  The two measures are not polynomially related in general: if $u_0=x$ and
$u_{j+1}=f(u_j,u_j)$, then $u_j$ has tree size $2^{j+1}-1$ but only $j+1$ distinct subterms.
All polynomial-time statements below use the shared-DAG measure.
\end{definition}

\begin{remark}[Why we fix an explicit size measure]
Whenever we claim an algorithm runs in polynomial time, the underlying input length is measured by
$\instsize{\Gamma}$ (or $\instsize{\mathbf t}$ for dispersion instances).
Input term tuples are assumed to be supplied as shared term DAGs; a syntax tree is a special case.
\end{remark}

\begin{definition}[Interpretation, solutions, and maximum code size]
Fix $n\ge 1$ and an alphabet $\A$ of size $n$.
An \emph{interpretation} $\mathcal I$ assigns to each $k$-ary $f\in\Fun$ a total function
$f^\mathcal I:\A^k\to\A$.
Given an assignment $\mathbf a\in\A^{\Vars}$, each term $u$ evaluates to a value
$u^{\mathcal I}(\mathbf a)\in\A$ by recursion on $u$.
An assignment $\mathbf a$ satisfies an equation $s=t$ if $s^{\mathcal I}(\mathbf a)=t^{\mathcal I}(\mathbf a)$.

The solution set $\Sol_\mathcal I(\Gamma;n)\subseteq \A^{\Vars}$ consists of all assignments of values in
$\A$ to $\Vars$ that satisfy every equation in $\Gamma$ under $\mathcal I$.
The \emph{maximum code size} is
\[
\Sn(\Gamma)\ :=\ \max_{\mathcal I} |\Sol_\mathcal I(\Gamma;n)|.
\]
\end{definition}

\subsection{Entropy normalisation (intuition)}
If we choose a solution uniformly at random from $\Sol_\mathcal I(\Gamma;n)$ and write the induced random
variables as $(X_x)_{x\in\Vars}$, then
\[
\log_n |\Sol_\mathcal I(\Gamma;n)| \;=\; H(X_{\Vars}),
\]
where $H$ is Shannon entropy normalized to base $n$.
Each equation of the form $z=f(u_1,\dots,u_k)$ enforces a deterministic dependence
$H(X_z\mid X_{u_1},\dots,X_{u_k})=0$.
This motivates the graph-entropy viewpoint.

\smallskip
\noindent
\textbf{Why the equality holds.}
If a random variable is uniformly distributed on a finite set of size $M$, then its Shannon entropy is $\log M$
(in whatever logarithmic base is used).
Thus, when we normalize logarithms to base $n$, uniform choice from the solution set gives
$H(X_{\Vars})=\log_n|\Sol_\mathcal I(\Gamma;n)|$.

\smallskip
\noindent
\textbf{Why term equations correspond to zero conditional entropy.}
An equation $z=f(u_1,\dots,u_k)$ asserts that $X_z$ is a deterministic function of $(X_{u_1},\dots,X_{u_k})$,
hence $H(X_z\mid X_{u_1},\dots,X_{u_k})=0$.
Conversely, any deterministic dependence of $X_z$ on $(X_{u_1},\dots,X_{u_k})$ can be represented by introducing an
appropriate $k$-ary function symbol and the corresponding equation.

\smallskip
\noindent
We emphasise that the results in this paper are stated and proved combinatorially; the entropy language merely
provides intuition for why max-flow/min-cut certificates appear in dispersion problems.

\section{Preprocessing: normalisation, quotienting, and diversification}
\label{sec:preprocess}

This section describes the pipeline that turns an arbitrary instance into a canonical graph-based
object.  The transformations are standard, but we state them explicitly because later sections use them
as a black box.  Concretely, we will use:
\begin{enumerate}[label=(\roman*)]
  \item flattening to normal form (Proposition~\ref{prop:flatten});
  \item quotienting explicit equalities (Lemma~\ref{lem:quotient});
  \item \emph{collision quotienting} to enforce unique left-hand sides (Lemma~\ref{lem:collision-quotient});
  \item (for the graph-based viewpoint) restricting attention to functional instances (Definition~\ref{def:fnf}),
        which are the natural format for deterministic dependency graphs and include all dispersion encodings; and
  \item diversification (Subsection~\ref{subsec:div}).
\end{enumerate}
Steps (i)--(iii) preserve $\Sn(\Gamma)$ for every alphabet size $n$ (up to the obvious projection on auxiliary
variables).  Diversification does \emph{not} preserve $S_n(\Gamma)$ in general; instead it provides a monotone upper bound
(Lemma~\ref{lem:div-easy}) and, for collision-free functional instances, a converse bound after shrinking the
alphabet by a factor $v=|\Vars|$ (Theorem~\ref{thm:sandwich}).

\begin{remark}[Effectiveness and size blow-up]
All preprocessing steps above are effective and can be carried out in time polynomial in the size of the
syntactic description of $\Gamma$.
Flattening introduces at most one auxiliary variable per subterm occurrence,
the quotienting steps (Lemmas~\ref{lem:quotient} and \ref{lem:collision-quotient}) can be implemented with a
union--find data structure, and diversification duplicates each defining equation once.
In particular, the dependency graph $G_\Gamma$ can be constructed in polynomial time from $\Gamma$.
\end{remark}

\subsection{Normal form (flattening)}
A term-coding instance is in \emph{normal form} if every non-equality equation has the shape
\begin{equation}
\label{eq:normalform}
f(u_1,\dots,u_k)=v\qquad (k\ge 0),
\end{equation}
where $v\in\Vars$ is a variable, the $u_i$ are variables, and $f$ is a function symbol of arity $k$.
We allow $k=0$ (so $f$ is a \emph{constant} symbol and the equation is simply $f()=v$).

\begin{proposition}[Flattening to normal form]
\label{prop:flatten}
For every term-coding instance $\Gamma$ there exists an instance $\Gamma^{\mathrm{nf}}$ over the \emph{same}
signature whose set of variables is $X\cup Z$ (where $X$ are the original variables and $Z$ are fresh
auxiliary variables) such that:
\begin{enumerate}[label=(\alph*)]
\item every equation of $\Gamma^{\mathrm{nf}}$ is either a variable equality $u=v$ or has depth~$1$ and the form
      $f(u_1,\dots,u_k)=v$ with $u_1,\dots,u_k,v\in X\cup Z$; and
\item for every finite alphabet $\A$ of size $n$ and every interpretation $\mathcal I$ of the function symbols on $\A$,
      restriction to the original variables induces a bijection
      $\Sol_\mathcal I(\Gamma^{\mathrm{nf}};n)\to \Sol_\mathcal I(\Gamma;n)$.
\end{enumerate}
In particular, $\Sn(\Gamma)=\Sn(\Gamma^{\mathrm{nf}})$ for all $n$.
\end{proposition}

\begin{proof}
Process every term that occurs in $\Gamma$ from the leaves upward.
Whenever a (non-variable) subterm $s=f(s_1,\dots,s_k)$ occurs, introduce a fresh auxiliary variable $z_s$
(if it has not been introduced already), and add the normal-form equation
\[
f(z_{s_1},\dots,z_{s_k}) = z_s,
\]
where for a variable $x$ we set $z_x:=x$.
Finally, in every original equation $s=t$ of $\Gamma$, replace each maximal proper subterm $u$ by its
auxiliary variable $z_u$.
The resulting finite system $\Gamma^{\mathrm{nf}}$ is in normal form by construction.

Fix $n$, an alphabet $\A$ of size $n$, and an interpretation $\mathcal I$ on $\A$.
If $\alpha:X\to\A$ satisfies $\Gamma$, extend it to $\widehat\alpha:X\cup Z\to\A$ by setting
$\widehat\alpha(z_s)=s^{\mathcal I}(\alpha)$ for every auxiliary variable $z_s$.
Then $\widehat\alpha$ satisfies all added equations, and the rewritten original equations hold because every
occurrence of a subterm was replaced by a variable constrained to equal that subterm.

Conversely, let $\beta\in\Sol_\mathcal I(\Gamma^{\mathrm{nf}};n)$ and write $\beta_X=\beta|_X$.
By induction on term depth, the added equations force $\beta(z_s)=s^{\mathcal I}(\beta_X)$ for every auxiliary
variable $z_s$. Substituting this identity into the rewritten original equations shows that $\beta_X$ satisfies
$\Gamma$.
The two maps $\alpha\mapsto\widehat\alpha$ and $\beta\mapsto\beta|_X$ are inverse bijections.
\end{proof}

\subsection{Quotienting forced variable equalities}
Flattening may introduce explicit variable equalities (e.g.\ $x_i=x_j$) in addition to equations of the
shape \eqref{eq:normalform}.  These can be eliminated by quotienting variables.

\begin{lemma}[Quotienting variable equalities]
\label{lem:quotient}
Let $\Gamma$ be in normal form except that it may additionally contain equations of the form $x_i=x_j$
between variables.
Let $\sim$ be the equivalence relation on $\Vars$ generated by these equalities, and form the quotient
instance $\Gamma/\!\sim$ by replacing each variable by a chosen representative of its $\sim$-class and then
deleting all variable-equality equations.
Then for every $n$, every alphabet $\A$ of size $n$, and every interpretation $\mathcal I$ of the non-variable symbols on $\A$, restriction
to representatives induces a bijection
\[
\Sol_\mathcal I(\Gamma;n)\ \cong\ \Sol_\mathcal I(\Gamma/\!\sim;n),
\]
and hence $\Sn(\Gamma)=\Sn(\Gamma/\!\sim)$.
\end{lemma}

\begin{proof}
Let $R\subseteq \Vars$ be a set of representatives, one from each $\sim$-class.
Given any solution assignment $\alpha:\Vars\to\A$ of $\Gamma$, its restriction $\alpha|_R$ is a solution of
$\Gamma/\!\sim$ because every occurrence of a variable is replaced by its representative.
Conversely, any assignment $\beta:R\to\A$ extends uniquely to an assignment $\bar\beta:\Vars\to\A$ constant on
$\sim$-classes.  Since $\Gamma$ and $\Gamma/\!\sim$ have identical defining equations after replacing
variables by representatives, $\bar\beta$ satisfies $\Gamma$ if and only if $\beta$ satisfies $\Gamma/\!\sim$.
Thus restriction to $R$ is a bijection on solutions, and taking maxima over $\mathcal I$ yields the claim.
\end{proof}

\subsection{Eliminating left-hand collisions}
\label{subsec:collisions}
Normal form does not prevent the \emph{same} left-hand side term from occurring multiple times, e.g.
\(f(u_1,\dots,u_k)=v\) and \(f(u_1,\dots,u_k)=w\).
In the original instance these equations force $v=w$ for \emph{every} interpretation, but after
\emph{diversification} the two right-hand sides would become independent.
Since our lower-bound simulation in Theorem~\ref{thm:sandwich} relies on merging diversified symbols back
into a single symbol, we must quotient out these ``implicit'' equalities first.

\begin{definition}[Left-hand collisions and collision-free normal form]
\label{def:collisionfree}
Two normal-form equations
\(f(u_1,\dots,u_k)=v\) and \(f(u_1,\dots,u_k)=w\) with the same function symbol $f$ and the same ordered
argument tuple $(u_1,\dots,u_k)$ form a \emph{left-hand collision} (including the case $k=0$, i.e.\ a constant
symbol $f$).
An instance is \emph{collision-free} if it has no left-hand collision with $v\neq w$; equivalently, for each
pair $(f,(u_1,\dots,u_k))$ there is at most one right-hand side variable.
This is sometimes called a ``Term-DAG'' or ``no-collision'' condition.
\end{definition}

\begin{lemma}[Collision quotienting]
\label{lem:collision-quotient}
Let $\Gamma$ be a normal-form instance with no explicit variable-equality equations.
Let $\approx$ be the least equivalence relation on $\Vars$ closed under the following rule:
whenever $\Gamma$ contains equations
$f(u_1,\dots,u_d)=v$ and $f(u'_1,\dots,u'_d)=w$ with $u_i\approx u'_i$ for every $i$, declare
$v\approx w$.
Let $\Gamma^{\mathrm{cf}}:=\Gamma/\!\approx$ be the quotient instance obtained by identifying each
$\approx$-class (as in Lemma~\ref{lem:quotient}) and then deleting duplicate equations.
Then for every $n$, every finite alphabet $\A$ of size $n$, and every interpretation $\mathcal I$ on $\A$,
restriction to representatives induces a
bijection
\[\Sol_{\mathcal I}(\Gamma;n)\ \cong\ \Sol_{\mathcal I}(\Gamma^{\mathrm{cf}};n).\]
In particular, $\Sn(\Gamma)=\Sn(\Gamma^{\mathrm{cf}})$, and $\Gamma^{\mathrm{cf}}$ is collision-free.
\end{lemma}

\begin{proof}
Every satisfying assignment $\alpha$ is constant on each $\approx$-class.  Indeed, this follows by induction
over the generation of $\approx$: if $u_i\approx u'_i$ for every $i$, then
$\alpha(u_i)=\alpha(u'_i)$ for every $i$, and therefore
\[
\alpha(v)=f^{\mathcal I}(\alpha(u_1),\dots,\alpha(u_d))
=f^{\mathcal I}(\alpha(u'_1),\dots,\alpha(u'_d))=\alpha(w).
\]
Thus any solution of $\Gamma$ is uniquely determined by its restriction to a set of class representatives, and
this restriction satisfies $\Gamma^{\mathrm{cf}}$.
Conversely, any solution of $\Gamma^{\mathrm{cf}}$ extends uniquely to a solution of $\Gamma$ by making it
constant on $\approx$-classes; duplicate equations in $\Gamma$ are redundant under such an assignment.
Finally, two quotient equations with the same left-hand side lift to equations of $\Gamma$ with
argumentwise $\approx$-equivalent tuples.  Their outputs are therefore $\approx$-equivalent, so the equations
coincide after deduplication.  Hence $\Gamma^{\mathrm{cf}}$ is collision-free.
\end{proof}

\begin{remark}[Why closure is necessary]
The collisions present before quotienting do not suffice.  For example,
$f(x)=v$, $f(x)=w$, $g(v)=a$, and $g(w)=b$ first force $v\approx w$; after that identification the last two
equations form a new collision.  The closed relation above also identifies $a$ and $b$.
\end{remark}

\subsection{Sources and functional normal form}
After flattening and quotienting, it is convenient to distinguish \emph{source} variables (inputs) from
\emph{defined} variables (computed by an equation).

\begin{definition}[Defined and source variables]
For a normal-form instance $\Gamma$ on variables $\Vars$, let
\[
\Def(\Gamma):=\{\,v\in\Vars : \text{there exists an equation of the form $f(u_1,\dots,u_k)=v$ in $\Gamma$}\,\}
\]
and $\Src(\Gamma):=\Vars\setminus \Def(\Gamma)$.
\end{definition}

\begin{definition}[Functional normal form (FNF)]
\label{def:fnf}
A normal-form instance $\Gamma$ is in \emph{functional normal form (FNF)} if every non-source variable is
defined by \emph{exactly one} equation; i.e.\ for each $v\in\Def(\Gamma)$ there is a unique equation of the
form $f(u_1,\dots,u_k)=v$ in $\Gamma$.
\end{definition}

\begin{remark}[Why FNF?]
The cleanest reduction from term-coding constraints to a single guessing game assigns one local update rule
to each non-source variable.  This is exactly the FNF condition.
Dispersion and most network-coding encodings naturally yield FNF instances.
When a preprocessing step produces multiple defining equations for the same variable, one can either
(i) treat the system as a conjunction of local constraints (leading to a slightly richer game model), or
(ii) introduce additional variables to separate definitions and work with a functional core.
Option (ii) can be implemented in polynomial time: give each defining equation its own fresh output variable,
so that the functional core is in FNF, and keep the equalities identifying each copy with its original variable
as explicit side constraints.  Restriction to the original variables is then a bijection between the solution
sets, so $S_n(\Gamma)$ is preserved for every $n$.  Quotienting the copy equalities away would merely restore
the multiple definitions.
For the complexity results in Sections~\ref{sec:dispersion}--\ref{sec:complexity-dichotomy}, we stay within the FNF
setting.
\end{remark}

\begin{definition}[Collision-free functional normal form (CFNF)]
\label{def:cfnf}
A normal-form instance $\Gamma$ is in \emph{collision-free functional normal form (CFNF)} if it is in FNF
(Definition~\ref{def:fnf}) and collision-free (Definition~\ref{def:collisionfree}).
Equivalently, $\Gamma$ is CFNF if it is FNF and has no left-hand collisions.
\end{definition}

\subsection{Diversification}
\label{subsec:div}
Let $\Gamma$ be a normal-form instance.  Its \emph{diversification} $\Gamma^{\mathrm{div}}$ is obtained by
replacing each occurrence of a function symbol by a fresh symbol of the same arity (one new symbol per
equation occurrence).  The resulting instance has the same variables but a larger signature.

\begin{lemma}[Diversification increases solutions]
\label{lem:div-easy}
For every $n$ we have $\Sn(\Gamma)\le \Sn(\Gamma^{\mathrm{div}})$.
\end{lemma}

\begin{proof}
Given an interpretation $\mathcal I$ of $\Gamma$, interpret every diversified copy $f^{(i)}$ as the same
function $f^\mathcal I$.  Then every solution of $\Gamma$ is a solution of $\Gamma^{\mathrm{div}}$.
Maximizing over interpretations yields the claim.
\end{proof}

\subsection{Dependency graph}
\begin{definition}[Dependency graph with sources]
Let $\Gamma$ be an FNF instance in normal form.
Its \emph{dependency graph} $G_\Gamma$ is the directed graph with vertex set $\Vars$, and with an edge
$u\to v$ for each equation $f(u_1,\dots,u_k)=v$ in $\Gamma$ and each argument variable $u\in\{u_1,\dots,u_k\}$.
We also record the designated set of sources $\Src(\Gamma)$.
\end{definition}

Intuitively, a non-source variable $v$ is \emph{functionally determined} by its in-neighbours in $G_\Gamma$.

\section{Guessing games and the guessing-number sandwich}
\label{sec:guess}

Guessing games were introduced in the study of information flow on graphs and are now standard in network
coding and index coding \parencite{riis2006information,riis2007graph,gadouleau2011graph}.

\subsection{Guessing games with sources}
Let $G=(V,E)$ be a finite directed graph and let $S\subseteq V$ be a set of \emph{sources}.
Fix an alphabet $\A$ of size $n$.
A configuration is an assignment $\mathbf a\in \A^{V}$.
Each vertex $v\in V\setminus S$ sees the values on its in-neighbours $N^-(v)$ and must output a guess for its
own value.  Formally, a (deterministic) strategy consists of local functions
\[
g_v:\A^{N^-(v)}\to \A\qquad (v\in V\setminus S).
\]
A configuration $\mathbf a$ is \emph{winning} if $g_v(\mathbf a|_{N^-(v)})=a_v$ for every $v\in V\setminus S$.
Let $W(G,S;n)$ denote the maximum number of winning configurations over all strategies.

\begin{definition}[Guessing number]
The (normalized) guessing number is
\[
\Guess(G,S;n)\ :=\ \log_n W(G,S;n).
\]
\end{definition}

\begin{remark}[Sources and self-loops]
Some formulations of guessing games work only with loopless digraphs and do not distinguish a set of
sources.  Our model with sources can be converted to the usual ``no-sources'' setting by adding an
identity self-loop at each $s\in S$ (so $s$ sees its own value and can guess it trivially), and conversely.
We keep sources explicit because they match the input variables in term-coding instances.
\end{remark}

\subsection{From diversified FNF instances to guessing games}
If $\Gamma$ is diversified, normal-form, and in FNF, then each equation
$f(u_1,\dots,u_k)=v$ can be read as a local rule that computes $v$ from its in-neighbours.
Thus interpretations of $\Gamma^{\mathrm{div}}$ correspond exactly to guessing strategies on $G_\Gamma$, and
solutions correspond exactly to winning configurations.

\begin{proposition}[Solutions = winning configurations for diversified FNF]
\label{prop:solutions-winning}
Let $\Gamma$ be an FNF instance in normal form, and let $\Gamma^{\mathrm{div}}$ be its diversification.
Then for every $n$,
\[
\Sn(\Gamma^{\mathrm{div}})\ =\ W(G_\Gamma,\Src(\Gamma);n),
\]
and hence $\log_n \Sn(\Gamma^{\mathrm{div}})=\Guess(G_\Gamma,\Src(\Gamma);n)$.
\end{proposition}

\begin{proof}
Each equation in $\Gamma^{\mathrm{div}}$ has a unique left-hand function symbol and defines the value of its
right-hand side variable from the argument variables.
Interpreting each such symbol is the same as choosing the corresponding local guessing function $g_v$.
An assignment to $\Vars$ satisfies all equations iff every defined variable equals the value prescribed by
its local function on the in-neighbour values; this is exactly the winning condition.
\end{proof}

\subsection{The sandwich bound}
The next theorem is the technical bridge between term coding and guessing games.
The \emph{upper} bound is immediate from the definition of diversification.
The \emph{lower} bound is a block-encoding argument in which we reserve disjoint \emph{regions} of the alphabet
for each variable, and then define every original function symbol piecewise on those regions.
This is exactly where the collision quotienting of Section~\ref{subsec:collisions} is needed.

\begin{theorem}[Diversification sandwich bounds]
\label{thm:sandwich}
Let $\Gamma$ be a collision-free normal-form instance with $v:=|\Vars|\ge 1$ variables, and let
$\Gamma^{\mathrm{div}}$ be its diversification.
Then for every alphabet size $n$,
\[
\Sn(\Gamma)\ \le\ \Sn(\Gamma^{\mathrm{div}}).
\]
Moreover, for every $n\ge v$ and $m:=\lfloor n/v\rfloor$,
\[
\Sn(\Gamma)\ \ge\ S_m(\Gamma^{\mathrm{div}}).
\]
\end{theorem}

\begin{proof}
\emph{Upper bound.}
Fix $n$, a finite alphabet $\A$ of size $n$, and an interpretation $\mathcal I$ of the symbols of $\Gamma$ on $\A$.
Define an interpretation $\mathcal I^{\mathrm{div}}$ of $\Gamma^{\mathrm{div}}$ on the same alphabet by
interpreting every diversified symbol $f^{(e)}$ as the original symbol $f$.
Then every equation of $\Gamma$ is literally one of the equations of $\Gamma^{\mathrm{div}}$ under
$\mathcal I^{\mathrm{div}}$, so
$\Sol_{\mathcal I}(\Gamma;n)\subseteq \Sol_{\mathcal I^{\mathrm{div}}}(\Gamma^{\mathrm{div}};n)$.
Taking cardinalities and then maximizing over $\mathcal I$ gives
$\Sn(\Gamma)\le \Sn(\Gamma^{\mathrm{div}})$.

\smallskip
\emph{Lower bound.}
Fix $n\ge v$ and set $m=\lfloor n/v\rfloor$.
For notational convenience identify a size-$n$ alphabet with $[n]$ and a size-$m$ alphabet with $[m]$.
Enumerate the variables as $\Vars=\{x_1,\dots,x_v\}$.
Choose pairwise disjoint blocks $B_1,\dots,B_v\subseteq [n]$ with $|B_i|=m$, and fix bijections
$\iota_i:[m]\to B_i$.
Let $\mathcal J$ be any interpretation of the diversified symbols on $[m]$.
We build an interpretation $\mathcal I$ of the original symbols on $[n]$ as follows.

Let $f$ be an original function symbol of arity $d$.
For each equation $e$ of $\Gamma$ of the form
$f(u_1,\dots,u_d)=v$ (so $u_1,\dots,u_d,v\in\Vars$), write $f^{(e)}$ for its diversified copy.

If $d=0$ (so $f$ is a constant symbol), then by the collision-free hypothesis there is at most
one equation $e:f()=x_j$ in $\Gamma$.
If such an equation exists, define the unique value of $f^{\mathcal I}:[n]^0\to[n]$ by
\[ f^{\mathcal I}() := \iota_j\bigl((f^{(e)})^{\mathcal J}()\bigr). \]
Otherwise define $f^{\mathcal I}():=0$.

For $d\ge 1$, define $f^{\mathcal I}:[n]^d\to [n]$ by the rule
\[
 f^{\mathcal I}(a_1,\dots,a_d)
 :=
 \begin{cases}
 \iota_j\Bigl((f^{(e)})^{\mathcal J}(\iota_{i_1}^{-1}(a_1),\dots,\iota_{i_d}^{-1}(a_d))\Bigr)
 &\text{if }\substack{\text{there is an equation }e: f(x_{i_1},\dots,x_{i_d})=x_j\\
                     a_\ell\in B_{i_\ell}\text{ for all }\ell},\\
 0 &\text{otherwise.}
 \end{cases}
\]
This is well-defined: if two equations $e,e'$ both match a tuple $(a_1,\dots,a_d)$, then for each coordinate
$\ell$ we would have $a_\ell\in B_{i_\ell}\cap B_{i'_\ell}$, forcing $i_\ell=i'_\ell$.
Thus $e$ and $e'$ have the same left-hand side, and since $\Gamma$ is collision-free we must have $e=e'$.
(We may define the ``otherwise'' value arbitrarily; the specific choice is irrelevant.)

Now take any solution $\alpha\in\Sol_{\mathcal J}(\Gamma^{\mathrm{div}};m)$.
Define an assignment $\widehat\alpha:\Vars\to [n]$ by $\widehat\alpha(x_i)=\iota_i(\alpha(x_i))$.
We claim $\widehat\alpha\in\Sol_{\mathcal I}(\Gamma;n)$.
Indeed, let $e$ be any equation of $\Gamma$.
If $e$ has arity $d\ge 1$, we can write it as $e:f(x_{i_1},\dots,x_{i_d})=x_j$.
Since $\widehat\alpha(x_{i_\ell})\in B_{i_\ell}$ for all $\ell$, the defining case of $f^{\mathcal I}$ applies and gives
\[
 f^{\mathcal I}(\widehat\alpha(x_{i_1}),\dots,\widehat\alpha(x_{i_d}))
 = \iota_j\Bigl((f^{(e)})^{\mathcal J}(\alpha(x_{i_1}),\dots,\alpha(x_{i_d}))\Bigr)
 = \iota_j\bigl(\alpha(x_j)\bigr)
 = \widehat\alpha(x_j),
\]
where we used that $\alpha$ satisfies the diversified equation $f^{(e)}(x_{i_1},\dots,x_{i_d})=x_j$.
If $e$ has arity $0$, then $e$ has the form $f()=x_j$ and similarly
$f^{\mathcal I}()=\iota_j((f^{(e)})^{\mathcal J}())=\iota_j(\alpha(x_j))=\widehat\alpha(x_j)$.
Thus $\widehat\alpha$ satisfies every equation of $\Gamma$.
The map $\alpha\mapsto\widehat\alpha$ is injective because each $\iota_i$ is injective, so
$|\Sol_{\mathcal I}(\Gamma;n)|\ge |\Sol_{\mathcal J}(\Gamma^{\mathrm{div}};m)|$.
Maximizing over $\mathcal J$ yields $\Sn(\Gamma)\ge S_m(\Gamma^{\mathrm{div}})$.
\end{proof}

\begin{corollary}[Asymptotic consequence]
\label{cor:sandwich-asympt}
Let $\Gamma$ be CFNF with $v=|\Vars|\ge 1$.
Then the two-sided bounds in Theorem~\ref{thm:sandwich} imply
\[
\limsup_{n\to\infty} \log_n \Sn(\Gamma)
\;=\;
\limsup_{n\to\infty} \log_n \Sn(\Gamma^{\mathrm{div}})
\;=\;
\limsup_{n\to\infty} \Guess(G_\Gamma,\Src(\Gamma);n).
\]
\end{corollary}

\begin{proof}
The first inequality in Theorem~\ref{thm:sandwich} gives
$\log_n \Sn(\Gamma)\le \log_n \Sn(\Gamma^{\mathrm{div}})$ for every $n$.
For the reverse bound on the limsup, apply the second inequality with $n=vm$ to get
$S_{vm}(\Gamma)\ge S_m(\Gamma^{\mathrm{div}})$.
Hence
\[
\log_{vm} S_{vm}(\Gamma)\ \ge\ \log_{vm} S_m(\Gamma^{\mathrm{div}})
\ =\ \frac{\log_m S_m(\Gamma^{\mathrm{div}})}{\log(vm)/\log m}.
\]
Taking $\limsup_{m\to\infty}$ and noting that $\log(vm)/\log m\to 1$, we obtain
\(
\limsup_{n\to\infty} \log_n \Sn(\Gamma)
\ge
\limsup_{m\to\infty} \log_m S_m(\Gamma^{\mathrm{div}}).
\)
Finally, Proposition~\ref{prop:solutions-winning} identifies
$\log_n \Sn(\Gamma^{\mathrm{div}})=\Guess(G_\Gamma,\Src(\Gamma);n)$ for every $n$.
\end{proof}

The three sequences in Corollary~\ref{cor:sandwich-asympt} in fact converge, so the limsups are limits;
this follows from the product and monotonicity argument for general term coding
\parencite{riis2026termcoding}.

\begin{remark}
The collision-free hypothesis in Theorem~\ref{thm:sandwich} is essential for the lower bound.
If $\Gamma$ contains both $f(\mathbf u)=v$ and $f(\mathbf u)=w$ with $v\ne w$, then in the original system we
have the implicit constraint $v=w$, while in $\Gamma^{\mathrm{div}}$ the two copies of $f$ can behave
independently.
Lemma~\ref{lem:collision-quotient} shows how to quotient variables so that these implicit equalities are made
explicit before diversification.
\end{remark}

\section{Dispersion and term-cut exponents}
\label{sec:dispersion}

Dispersion isolates the image-size aspect of information flow and was introduced and analysed in
\textcite{riis2019max}.  We recall the definitions and the algorithmic consequences needed for the complexity analysis.

\subsection{Dispersion instances}
\begin{definition}[Dispersion instance]
\label{def:dispersion-multisorted}
Fix $k,r\ge 1$.
A \emph{dispersion instance} is a tuple of terms
$\mathbf t(\mathbf x)=(t_1(\mathbf x),\dots,t_r(\mathbf x))$ in input variables
$\mathbf x=(x_1,\dots,x_k)$ over a signature $\Fun$.
For an alphabet $\A$ of size $n$ and interpretation $\mathcal I$, it induces a map
$\Theta^\mathcal I:\A^k\to \A^r$ as in Section~\ref{sec:intro}.
The dispersion at alphabet size $n$ is
\[
\Disp_n(\mathbf t)\ :=\ \max_{\mathcal I} |\Img(\Theta^\mathcal I)|.
\]
\end{definition}

Throughout the remainder of the paper, $k,r\ge 1$ and every alphabet has size $n\ge 2$, unless a statement
explicitly says otherwise.  Dispersion depends only on the cardinality of the alphabet: transport an
interpretation along any bijection between two alphabets of the same size.  We may therefore identify an
$n$-element alphabet with $[n]$ without loss of generality.

\begin{remark}[Degenerate parameters]
The restriction $n\ge 2$ is essential for the decision problems: on a singleton alphabet,
$\Disp_1(\mathbf t)=1=1^r$ for every instance.  If zero outputs were allowed, every instance would be
perfect, whereas with no inputs and at least one output, perfect dispersion is impossible on an alphabet of
size at least two.  Our standing assumptions remove these uninformative cases.
\end{remark}

\subsection{The single-function diamond: full rate without perfection}
\label{sec:diamond}
Example~\ref{ex:diamond-intro} already shows, by a one-line pigeonhole argument, that certain natural
single-symbol dispersion gadgets cannot achieve perfect dispersion even though their term-cut exponent is full.
We restate the gadget here and isolate the source of the obstruction.

\begin{example}[Single-function diamond]
\label{ex:diamond}
Let $x,y,z,w\in\A$ be inputs and let $f$ be a single binary function symbol.
The dispersion instance is
\begin{equation}
\label{eq:diamond-terms}
\mathbf t(x,y,z,w)\ :=\ \bigl( f(x,y),\, f(x,z),\, f(y,w),\, f(z,w)\bigr).
\end{equation}
\end{example}

\begin{lemma}[Perfect dispersion fails for the diamond gadget]
\label{lem:diamond-not-bij}
For every $n\ge 2$ and every interpretation $f^\mathcal I:\A^2\to\A$, the induced map
$\Theta^\mathcal I:\A^4\to\A^4$ of \eqref{eq:diamond-terms} is not injective (hence not bijective/surjective).
\end{lemma}

\begin{proof}
Restrict to inputs with $y=z$.
There are $n^3$ such inputs $(x,y,z,w)=(x,y,y,w)$.
But for these inputs the output has the repeated form
$(f(x,y),f(x,y),f(y,w),f(y,w))=(a,a,b,b)$, and hence ranges over at most $n^2$ values.
Since $n^3>n^2$ for $n\ge 2$, two distinct inputs map to the same output.
\end{proof}

\begin{remark}[Exact obstruction versus asymptotic rate]
At the level of coarse rates, the diamond gadget has the full integer exponent $D(\mathbf t)=4$
(Theorem~\ref{thm:disp-exponent} together with the explicit min-cut computation in Lemma~\ref{lem:diamond-maxflow}).
Nevertheless, Lemma~\ref{lem:diamond-not-bij} shows that the exact ``perfect'' target $n^4$ is unreachable.
Thus the same instance has full asymptotic rate but an empty perfect-alphabet spectrum: the exponent records
coarse growth while omitting the finite-alphabet injectivity obstruction.
\end{remark}

\begin{remark}[Contrast: diversification removes the obstruction]
If one replaces the single symbol $f$ in \eqref{eq:diamond-terms} by four independent symbols
$f_{1},f_{2},f_{3},f_{4}$, then the resulting map can be made bijective by interpreting
$f_{1}(x,y)=x$, $f_{2}(x,z)=z$, $f_{3}(y,w)=y$, and $f_{4}(z,w)=w$.
In other words, it is the \emph{reuse} of the same local function that creates the global obstruction.
\end{remark}

\subsection{Embedding dispersion into term coding}
Dispersion can be expressed as a term-coding maximisation problem by adding \emph{decoder} function symbols.
Let $\mathbf y=(y_1,\dots,y_r)$ be fresh variables and introduce fresh function symbols
$\mathbf h=(h_1,\dots,h_k)$ where each $h_j$ is $r$-ary.

Consider the term-coding instance $\Gamma_{\mathbf t}$ on variables $\mathbf x,\mathbf y$ with equations
\begin{equation}
\label{eq:disp-to-tc}
y_i = t_i(\mathbf x)\quad (i\in[r]),
\qquad
x_j = h_j(\mathbf y)\quad (j\in[k]).
\end{equation}

Figure~\ref{fig:disp-embed} illustrates the idea of this embedding.

\begin{figure}[t]
\centering
\begin{tikzpicture}[>=Stealth, node distance=16mm and 14mm, every node/.style={font=\small}]
\node (x) {$\mathbf x\in\A^k$};
\node[draw,rounded corners, right=of x, inner sep=3pt] (t) {$\mathbf t$};
\node[right=of t] (y) {$\mathbf y\in\A^r$};

\node[draw,rounded corners, below=10mm of t, inner sep=3pt] (h) {$\mathbf h$};

\draw[->] (x) -- node[above] {$\mathbf y=\mathbf t(\mathbf x)$} (t);
\draw[->] (t) -- (y);
\draw[->] (y) -- node[right] {$\mathbf x=\mathbf h(\mathbf y)$} (h);
\draw[->] (h.west) -| (x.south);

\end{tikzpicture}
\caption{Dispersion as term coding via decoder symbols (Lemma~\ref{lem:disp-embed}).  Choosing $\mathbf h$ selects one preimage for each output $\mathbf y$ in the image of $\mathbf t$.}
\label{fig:disp-embed}
\end{figure}

\begin{lemma}[Dispersion equals maximum code size of the embedding]
\label{lem:disp-embed}
For every $n\ge 1$,
\[
\Disp_n(\mathbf t)\ =\ \Sn(\Gamma_{\mathbf t}).
\]
On the right, the maximum is over interpretations of the enlarged signature: the original symbols and all
fresh decoder symbols $h_j$ are chosen freely.
\end{lemma}

\begin{proof}
Fix an alphabet $\A$ of size $n$ and an interpretation $\mathcal I$ of the symbols in $\mathbf t$.
Let $Y:=\Img(\Theta^\mathcal I)\subseteq \A^r$.
Choose for each $\mathbf y\in Y$ one preimage $\mathbf x(\mathbf y)\in \A^k$ such that
$\Theta^\mathcal I(\mathbf x(\mathbf y))=\mathbf y$, and define $h_j^\mathcal I(\mathbf y):=x_j(\mathbf y)$
for $\mathbf y\in Y$ (and define $h_j^\mathcal I$ arbitrarily on $\A^r\setminus Y$).
Then the solutions of \eqref{eq:disp-to-tc} are exactly the pairs
$(\mathbf x(\mathbf y),\mathbf y)$ with $\mathbf y\in Y$, hence $|\Sol_\mathcal I(\Gamma_{\mathbf t};n)|=|Y|$.
Maximizing over $\mathcal I$ yields $\Sn(\Gamma_{\mathbf t})\ge \Disp_n(\mathbf t)$.

Conversely, for any interpretation $\mathcal I$ of all symbols (including $\mathbf h$), every solution
$(\mathbf x,\mathbf y)$ satisfies $\mathbf y=\Theta^\mathcal I(\mathbf x)$, so $\mathbf y$ lies in the image
of $\Theta^\mathcal I$.  Moreover, by the equations $x_j=h_j(\mathbf y)$, each $\mathbf y$ determines at most
one $\mathbf x$.  Hence the number of solutions is at most the image size:
$|\Sol_\mathcal I(\Gamma_{\mathbf t};n)|\le |\Img(\Theta^\mathcal I)|$.
Taking maxima yields $\Sn(\Gamma_{\mathbf t})\le \Disp_n(\mathbf t)$.
\end{proof}

The limit defining $D(\mathbf t)$ exists independently of its min-cut description.  Indeed,
Lemma~\ref{lem:disp-embed} identifies dispersion with the maximum solution count of a term-coding instance,
and the general term-coding convergence theorem applies to that instance
\parencite{riis2026termcoding}.  The term-cut theorem below adds integrality, an exact combinatorial formula,
and polynomial-time computability.

\subsection{Integer exponent and polynomial-time computability}
The following term-cut formulation is due to \textcite{riis2019max}.

\begin{definition}[Term-cut]
\label{def:term-cut}
Let $\Gamma_{\mathrm{sub}}(\mathbf t)$ be the set of distinct subterms of
$t_1,\dots,t_r$, including the input variables, and put
$s:=|\Gamma_{\mathrm{sub}}(\mathbf t)|$.  A \emph{term-cut} is a set
$C\subseteq\Gamma_{\mathrm{sub}}(\mathbf t)$ such that every output term $t_i$ can be built by applying
function symbols to elements of $C$; constants require no cut element.  Let
\[
\rho(\mathbf t):=\min\{\,|C|:C\text{ is a term-cut for }\mathbf t\,\}.
\]
Equivalently, $\rho(\mathbf t)$ is the minimum vertex cut separating the input variables from the output
terms in the shared term DAG when every subterm node, including every input variable, has unit capacity.
By vertex splitting, it is computable by a standard max-flow computation.
\end{definition}

The definition is syntactic: it counts distinct subterms, not occurrences.  Reuse of a function symbol is
handled in the lower-bound interpretation below by tagging alphabet values with subterms, so all occurrences
still use one common operation table.

\begin{theorem}[Integral dispersion exponent {\protect\citep{riis2019max}}]
\label{thm:disp-exponent}
Let $\mathbf t$ be a dispersion instance, let $s=|\Gamma_{\mathrm{sub}}(\mathbf t)|$, and put
$\rho=\rho(\mathbf t)$.  Then:
\begin{enumerate}[label=\textit{(\alph*)}, leftmargin=*]
\item $\Disp_n(\mathbf t)\le n^\rho$ for every $n\ge 1$;
\item $\Disp_{sm}(\mathbf t)\ge m^\rho$ for every $m\ge 1$, and consequently
      $\Disp_n(\mathbf t)\ge (2s)^{-\rho}n^\rho$ for every $n\ge 2s$.
\end{enumerate}
Therefore $\Disp_n(\mathbf t)=\Theta(n^\rho)$,
\[
D(\mathbf t)=\lim_{n\to\infty}\frac{\log\Disp_n(\mathbf t)}{\log n}=\rho(\mathbf t)\in\mathbb Z_{\ge0},
\]
and $D(\mathbf t)$ is computable in time polynomial in $\instsize{\mathbf t}$.
\end{theorem}

\begin{proof}
For \textit{(a)}, let $C=\{c_1,\dots,c_q\}$ be a term-cut.  Under any interpretation and input assignment,
each output value is determined by the tuple
$(c_1^\mathcal I,\dots,c_q^\mathcal I)$: this follows by structural induction on an expression of $t_i$ from
$C$.  The image of $\Theta^\mathcal I$ therefore has at most $n^q$ elements.  Minimise over $C$ and then
maximise over $\mathcal I$ to obtain $\Disp_n(\mathbf t)\le n^\rho$.

For \textit{(b)}, Menger's theorem \citep{menger1927} supplies $\rho$ pairwise vertex-disjoint paths
$P_1,\dots,P_\rho$ in the unit-node-capacity term DAG, with $P_a$ running from an input variable
$x_{v_a}$ to an output term $t_{i_a}$.  Let $B$ be an $m$-element set and use the tagged alphabet
\[
\A=\Gamma_{\mathrm{sub}}(\mathbf t)\times B,
\qquad |\A|=sm.
\]
Write $(u,\beta)$ for block value $\beta$ tagged by subterm $u$.  For each symbol $f$, define a single
operation table as follows.  Call an argument tuple
$(u_1,\beta_1),\dots,(u_d,\beta_d)$ well formed when
$u:=f(u_1,\dots,u_d)$ belongs to $\Gamma_{\mathrm{sub}}(\mathbf t)$.  On such a tuple, if $u$ lies on a path $P_a$,
output $(u,\beta_j)$, where $u_j$ is the predecessor of $u$ on $P_a$; vertex-disjointness makes $a$ unique.
Otherwise output $(u,\beta_0)$ for a fixed default $\beta_0\in B$, and define the table arbitrarily on
all other tuples.

Assign $(x_{v_a},b_a)$ to the source of $P_a$ and $(x_v,\beta_0)$ to every other input $x_v$.  Induction over
subterms shows that a subterm on $P_a$ evaluates with block value $b_a$.  Hence output $t_{i_a}$ records
$b_a$, and distinct tuples $(b_1,\dots,b_\rho)\in B^\rho$ give distinct output tuples.  Thus
$\Disp_{sm}(\mathbf t)\ge m^\rho$.  A single table serves every occurrence of $f$: the tags, rather than
diversification, route the information.

Finally, dispersion is monotone under enlargement of the alphabet.  Indeed, extend an interpretation from a
nonempty subalphabet by preserving every operation on tuples from that subalphabet and defining it arbitrarily
elsewhere.  For $n\ge2s$, take $m=\lfloor n/s\rfloor$; then $sm\le n$ and $m\ge n/(2s)$, which gives the
stated lower bound.  The two bounds imply the limit and its value $\rho$.  The min-cut is computable in
polynomial time by vertex splitting and max-flow.
\end{proof}

\begin{remark}[Origin and scope]
The term-cut invariant and routing argument are from \textcite{riis2019max}; the proof above records both
halves in the precise form needed here.  The only graph-theoretic input left implicit is Menger's theorem.
The R\'enyi-entropy results of that paper are not used.
\end{remark}

\subsection{Computing \texorpdfstring{$D(\mathbf t)$}{D(t)}: an explicit polynomial-time procedure}
\label{subsec:compute-exponent}

For completeness (and to make the polynomial-time side fully concrete), we spell out how
Theorem~\ref{thm:disp-exponent} yields an explicit algorithm for computing the integer exponent $D(\mathbf t)$.

Let $\instsize{\mathbf t}$ be the shared-DAG size fixed in Definition~\ref{def:instance-size}.
The term DAG has $O(\!\instsize{\mathbf t})$ nodes and edges.  Give every subterm node unit capacity, add a
supersource and supersink, and split each node into an in-node and an out-node joined by a unit-capacity edge.
The resulting ordinary flow network has polynomial size, and its minimum cut equals $\rho(\mathbf t)=D(\mathbf t)$.

Algorithm~\ref{alg:disp-exponent} gives the resulting deterministic procedure in pseudocode.  Its input is the
syntactic tuple $\mathbf t$ and its output is the integer $D(\mathbf t)$; the numbered lines specify the graph
construction and max-flow computation.

\begin{algorithm}[t]
\caption{Polynomial-time computation of $D(\mathbf t)$ by vertex-split max-flow.}
\label{alg:disp-exponent}
\begin{algorithmic}[1]
\raggedright
\State \textbf{Input:} a dispersion instance $\mathbf t(\mathbf x)=(t_1(\mathbf x),\dots,t_r(\mathbf x))$
\State \textbf{Output:} the integer exponent $D(\mathbf t)$
\State Construct the shared term DAG of $\mathbf t$: nodes are the distinct subterms, including input variables,
      and edges go from arguments to the subterm they form.
\State Give every subterm node unit capacity and apply standard vertex splitting; connect a supersource to the
      input variables and the output terms to a supersink with sufficiently large capacities.
\State Compute a maximum flow value $F$ from the supersource to the supersink in the vertex-split network.
\State \Return $F$ \Comment{$F=D(\mathbf t)$ by Theorem~\ref{thm:disp-exponent}.}
\end{algorithmic}
\end{algorithm}

\begin{remark}[Running time]
Since the vertex-split network has size polynomial in $\instsize{\mathbf t}$, any polynomial-time max-flow algorithm
computes $D(\mathbf t)$ in polynomial time.  This immediately yields a polynomial-time algorithm for computing $D(\mathbf t)$ and for the rate-threshold questions discussed in Section~\ref{sec:complexity-dichotomy}.
\end{remark}

\subsection{A worked min-cut computation: the diamond gadget}
\label{subsec:diamond-maxflow}

We now return to Example~\ref{ex:diamond} and verify explicitly that the term-cut value
(hence the dispersion exponent) is~$4$.

\begin{lemma}[Diamond gadget has min-cut $4$]
\label{lem:diamond-maxflow}
Let $\mathbf t=(f(x,y),f(x,z),f(y,w),f(z,w))$ be the diamond instance.
Then $\rho(\mathbf t)=4$, and consequently $D(\mathbf t)=4$.
\end{lemma}

\begin{proof}
The four variables $\{x,y,z,w\}$ form a term-cut, so $\rho(\mathbf t)\le4$.
Conversely, the term DAG contains the four vertex-disjoint paths
\[
  x\to u_1,\qquad z\to u_2,\qquad y\to u_3,\qquad w\to u_4,
\]
where $u_1=f(x,y)$, $u_2=f(x,z)$, $u_3=f(y,w)$, and $u_4=f(z,w)$.
A term-cut separates the input variables from the output terms in the term DAG
(Definition~\ref{def:term-cut}), so it contains a vertex of each of these paths; vertex-disjointness
forces at least four elements.  (This direction is weak duality and does not use Menger's theorem.)  Thus
$\rho(\mathbf t)=4$, and Theorem~\ref{thm:disp-exponent} gives $D(\mathbf t)=4$.
\end{proof}

The routing construction can also be seen directly.  Let
$\A=\{x,y,z,w\}\times B$ and fix one tag, say $x$, for outputs.  Define the single shared table by
\[
f((i,\beta),(j,\beta'))=
\begin{cases}
(x,\beta), &(i,j)\in\{(x,y),(y,w)\},\\
(x,\beta'),&(i,j)\in\{(x,z),(z,w)\},
\end{cases}
\]
with arbitrary values in all other cases.  On inputs tagged by their variable names, the four outputs recover,
in order, $\beta_x,\beta_z,\beta_y,\beta_w$.  Hence the image has at least
$|B|^4=(|\A|/4)^4$ elements: sharing $f$ costs a constant factor but not the exponent.

The corresponding simplified term graph is shown in Figure~\ref{fig:diamond-termgraph}.

\begin{figure}[t]
\centering
\begin{tikzpicture}[>=Latex, thick]
  % input variables
  \node[draw,circle,inner sep=2pt] (x)  at (0, 1.2) {$x$};
  \node[draw,circle,inner sep=2pt] (y)  at (0, 0.4) {$y$};
  \node[draw,circle,inner sep=2pt] (z)  at (0,-0.4) {$z$};
  \node[draw,circle,inner sep=2pt] (w)  at (0,-1.2) {$w$};

  % operation/output nodes (one per occurrence of f)
  \node[draw,circle,inner sep=2pt] (u1) at (2, 1.0) {$u_{1}$};
  \node[draw,circle,inner sep=2pt] (u2) at (2, 0.3) {$u_{2}$};
  \node[draw,circle,inner sep=2pt] (u3) at (2,-0.3) {$u_{3}$};
  \node[draw,circle,inner sep=2pt] (u4) at (2,-1.0) {$u_{4}$};

  % sink
  \node (t) at (4,0) {$t$};

  % edges into operation nodes
  \draw[->] (x) -- (u1);
  \draw[->] (y) -- (u1);

  \draw[->] (x) -- (u2);
  \draw[->] (z) -- (u2);

  \draw[->] (y) -- (u3);
  \draw[->] (w) -- (u3);

  \draw[->] (z) -- (u4);
  \draw[->] (w) -- (u4);

  % output-to-sink edges
  \draw[->] (u1) -- (t);
  \draw[->] (u2) -- (t);
  \draw[->] (u3) -- (t);
  \draw[->] (u4) -- (t);
\end{tikzpicture}
\caption{The term DAG for the diamond gadget.  Every circled subterm node, including each input variable,
has unit capacity; the sink is uncapacitated.  The four disjoint paths used in
Lemma~\ref{lem:diamond-maxflow} certify a cut value of $4$.}
\label{fig:diamond-termgraph}
\end{figure}

\section{Exact and asymptotic decision theory}
\label{sec:complexity-dichotomy}

We now focus on the complexity of analysing dispersion problems.
Throughout this paper we work in the \emph{single-sorted} setting, meaning that all variables range over a single finite alphabet $\A$ with $|\A|=n$.
(A separate paper treats multi-sorted term coding, where variables may range over different finite sets, and where dispersion admits a corresponding multi-sorted generalisation.)
To state our results precisely, we first define five explicit decision problems with clean quantifier structure.

\begin{definition}[Decision problems for dispersion]
\label{def:decision-problems}
Let $\mathbf t(\mathbf x) = (t_1(\mathbf x),\dots,t_r(\mathbf x))$ be a dispersion instance with $k$ input variables and $r$ output terms.

\begin{itemize}[leftmargin=*]
\item \textbf{Problem $\textsc{Fixed-Perfect-Disp}(\mathbf t, n)$:}\\
\textbf{Input:} A dispersion instance $\mathbf t$ and an integer $n \ge 2$.\\
\textbf{Question:} Is $\Disp_n(\mathbf t) = n^r$?

\item \textbf{Problem $\textsc{Perfect-Disp}_r(\mathbf t)$:}\\
\textbf{Input:} A dispersion instance $\mathbf t$ with $r$ output terms.\\
\textbf{Question:} Does there exist an integer $n \ge 2$ and an interpretation $\mathcal I$ on a domain $A$ with $|A| = n$ such that $\Theta^{\mathcal I}$ is surjective (equivalently, $\Disp_n(\mathbf t) = n^r$)?

\item \textbf{Problem $\textsc{Perfect-Disp}(\mathbf t)$:}\\
\textbf{Input:} An arbitrary dispersion instance $\mathbf t$, with the number $r$ of output terms part of the input.\\
\textbf{Question:} Does there exist $n\ge 2$ such that $\Disp_n(\mathbf t)=n^r$?

\item \textbf{Problem $\textsc{Full-Rate-Disp}(\mathbf t)$:}\\
\textbf{Input:} An arbitrary dispersion instance $\mathbf t$ with $r$ output terms.\\
\textbf{Question:} Is $D(\mathbf t)=r$?

\item \textbf{Problem $\textsc{Asymp-Threshold}_{d,h_d}(\mathbf t)$:}\\
\textbf{Fixed parameters:} An integer $d\ge 0$ and a threshold function
$h_d:\mathbb N\to\mathbb N$ satisfying \eqref{eq:threshold-conditions} below.\\
\textbf{Input:} A dispersion instance $\mathbf t$.\\
\textbf{Question:} Does there exist $N \in \mathbb{N}$ such that for all $n \ge N$:
$\Disp_n(\mathbf t) \ge h_d(n)$?
\end{itemize}
\end{definition}

For each fixed $n$, \textsc{Fixed-Perfect-Disp} is decidable by brute force over the finitely many
interpretations.  The content of the exact problems above is the unbounded existential quantification over
the alphabet size.

\begin{remark}[Input representation for $\textsc{Asymp-Threshold}$]
In Problem $\textsc{Asymp-Threshold}_{d,h_d}$, both $d$ and $h_d$ are fixed parameters of the problem
variant, not part of the input. For example, one may take $h_d(n)=n^d+1$ for a fixed $d$. This avoids
ambiguity about the encoding and evaluation complexity of an input threshold function.
\end{remark}

A key property distinguishing dispersion problems within the broader Term Coding framework relates to their asymptotic behaviour over an $n$-element domain $A$.
Recall from Theorem~\ref{thm:disp-exponent} that for any dispersion instance $\mathbf t$:
\[
\Disp_n(\mathbf t) = \Theta\!\left(n^{D(\mathbf t)}\right)
\]
where the exponent $D(\mathbf t) = \lim_{n\to\infty} \frac{\log \Disp_n(\mathbf t)}{\log n}$ is always an
\textbf{integer} and equals the minimum term-cut $\rho(\mathbf t)$ of the shared term DAG.
This integer exponent property is fundamental to the tractability of the asymptotic question.
For the exact question, square instances are exactly square-bijectivity instances, while padding gives
hardness reductions from fixed-dimensional square bijectivity to fixed-output perfect dispersion.  These facts do
not amount to a reduction of every rectangular perfect-dispersion instance to square bijectivity.

\subsection{Square bijectivity and perfect dispersion}
\label{subsec:square-bijectivity}

\begin{definition}[Square bijectivity]
\label{def:bij3}
Fix $k\ge 1$.
The problem $\textsc{BIJ}_k$ is:
\begin{quote}
\textbf{Input:} a tuple of $k$ terms
\[
\mathbf t(\mathbf x)=(t_1(\mathbf x),\dots,t_k(\mathbf x))
\]
in exactly the variables
\[
\mathbf x=(x_1,\dots,x_k).
\]
\textbf{Question:} does there exist a finite set $A$ with $|A|\ge 2$ and an interpretation $\mathcal I$ of the
function symbols such that the induced map
\[
\Theta^\mathcal I:A^k\to A^k,\qquad
\mathbf a\mapsto \bigl(t_1^\mathcal I(\mathbf a),\dots,t_k^\mathcal I(\mathbf a)\bigr)
\]
is bijective?
\end{quote}
The unrestricted problem $\textsc{BIJ}$ has the same input, except that $k$ is part of the input rather than a
fixed parameter.
\end{definition}

\begin{lemma}[Square perfect dispersion equals square bijectivity]
\label{lem:square-perfect-equals-bij}
Let $\mathbf t(\mathbf x)=(t_1(\mathbf x),\dots,t_k(\mathbf x))$ be a dispersion instance with exactly $k$ input
variables and exactly $k$ output terms.
Then the following are equivalent:
\begin{enumerate}[label=\textit{(\roman*)}, leftmargin=*]
\item there exists $n\ge 2$ such that
\[
\Disp_n(\mathbf t)=n^k;
\]
\item there exists a finite set $A$ with $|A|\ge 2$ and an interpretation $\mathcal I$ such that
\[
\Theta^\mathcal I:A^k\to A^k
\]
is surjective;
\item there exists a finite set $A$ with $|A|\ge 2$ and an interpretation $\mathcal I$ such that
\[
\Theta^\mathcal I:A^k\to A^k
\]
is bijective.
\end{enumerate}
\end{lemma}

\begin{proof}
The equivalence of \textit{(ii)} and \textit{(iii)} is immediate because $A^k$ is finite and the domain and codomain
have the same size.

\textit{(i)} $\Rightarrow$ \textit{(ii)}:
if $\Disp_n(\mathbf t)=n^k$, the maximum defining $\Disp_n(\mathbf t)$ is over the finitely many
interpretations on a fixed finite alphabet and is therefore attained.  Hence there exists an interpretation
$\mathcal I$ on an $n$-element set
$A$ such that
\[
|\Img(\Theta^\mathcal I)|=n^k=|A^k|.
\]
Hence $\Theta^\mathcal I$ is surjective.

\textit{(ii)} $\Rightarrow$ \textit{(i)}:
if $\Theta^\mathcal I$ is surjective on a finite set $A$ of size $n$, then
\[
|\Img(\Theta^\mathcal I)|=|A^k|=n^k.
\]
Therefore
\[
\Disp_n(\mathbf t)\ge n^k.
\]
Since the codomain has size exactly $n^k$, we always have $\Disp_n(\mathbf t)\le n^k$, and hence
\[
\Disp_n(\mathbf t)=n^k.
\]
\end{proof}

\begin{corollary}[The square restriction of perfect dispersion]
\label{cor:bij-reduces-perfect}
On the subfamily with the same number of inputs and outputs,
$\textsc{Perfect-Disp}$ is exactly $\textsc{BIJ}$.  Consequently, the identity transformation is a
polynomial-time many-one reduction
\[
\textsc{BIJ}\ \le_m\ \textsc{Perfect-Disp}.
\]
For every fixed $k$, the analogous square restriction of
$\textsc{Perfect-Disp}_k$ is exactly $\textsc{BIJ}_k$.
\end{corollary}

\begin{remark}[Inverse-extension form of square bijectivity]
\label{rem:inverse-extension}
For a square tuple
\[
\mathbf t(\mathbf x)=(t_1(\mathbf x),\dots,t_k(\mathbf x)),
\]
adjoin genuinely fresh $k$-ary symbols
\[
s_1(\mathbf y),\dots,s_k(\mathbf y)
\]
and require
\[
s_i\bigl(t_1(\mathbf x),\dots,t_k(\mathbf x)\bigr)=x_i,
\qquad
t_i\bigl(s_1(\mathbf y),\dots,s_k(\mathbf y)\bigr)=y_i
\quad (1\le i\le k).
\]
If an interpretation of the enlarged signature satisfies both families of equations identically, then the
restriction to the original signature makes $\Theta$ bijective and the $s_i$ form its inverse.  Conversely,
every interpretation for which $\Theta$ is bijective extends, without changing any original symbol, by
interpreting the $s_i$ as the coordinates of $\Theta^{-1}$.
Thus the exact side of the complexity landscape is fundamentally a finite satisfiability problem for a purely equational
inverse-extension theory.
\end{remark}

\begin{problem}[Square bijectivity for $k=3$]
\label{prob:bij3-open}
Is the problem $\textsc{BIJ}_3$ decidable?
\end{problem}

\begin{remark}[Status of Problem~\ref{prob:bij3-open}]
\label{rem:bij-status-brief}
To the best of the author's knowledge, the decidability status of unrestricted $\textsc{BIJ}$ and of
$\textsc{BIJ}_k$ for every fixed $k\ge 2$ remains open.  If $\textsc{BIJ}_3$ turns out to be undecidable, then
the reductions below immediately yield undecidability of perfect dispersion for every fixed $r\ge 3$.  The
converse direction is not established in this paper.  See Section~\ref{sec:bij-status} for a discussion of what
is currently known.
\end{remark}

\subsection{Perfect-alphabet spectra and full asymptotic rate}
\label{subsec:perfect-implies-rate}

\begin{definition}[Perfect-alphabet spectrum]
\label{def:perfect-spectrum}
For a tuple $\mathbf t$ with $r$ output terms, define
\[
\PerfSpec(\mathbf t)
:=
\{\,n\in\mathbb N:n\ge2\text{ and }\Disp_n(\mathbf t)=n^r\,\}.
\]
Thus $\textsc{Perfect-Disp}(\mathbf t)$ asks whether $\PerfSpec(\mathbf t)$ is nonempty.
\end{definition}

\begin{lemma}[Product alphabets preserve surjectivity]
\label{lem:product-surjectivity}
Let $\mathcal I$ and $\mathcal J$ be interpretations of the same signature on finite alphabets $A$ and $B$.
If the term maps $\Theta^{\mathcal I}:A^k\to A^r$ and
$\Theta^{\mathcal J}:B^k\to B^r$ are surjective, then there is an interpretation
$\mathcal I\times\mathcal J$ on $A\times B$ for which
\[
\Theta^{\mathcal I\times\mathcal J}:(A\times B)^k\longrightarrow(A\times B)^r
\]
is surjective.
\end{lemma}

\begin{proof}
For each $q$-ary symbol $f$, including $q=0$, define the operation on $A\times B$ coordinatewise by
\[
f^{\mathcal I\times\mathcal J}
  \bigl((a_1,b_1),\ldots,(a_q,b_q)\bigr)
=
\bigl(f^{\mathcal I}(a_1,\ldots,a_q),f^{\mathcal J}(b_1,\ldots,b_q)\bigr).
\]
A structural induction shows that every term is evaluated coordinatewise.  Under the canonical identifications
$(A\times B)^k\cong A^k\times B^k$ and $(A\times B)^r\cong A^r\times B^r$, the resulting term map is
$\Theta^{\mathcal I}\times\Theta^{\mathcal J}$, which is surjective.
\end{proof}

\begin{corollary}[Direct powers preserve surjectivity]
\label{lem:direct-power}
If an interpretation on $A$ makes $\Theta:A^k\to A^r$ surjective, then for every $m\ge1$ the same tuple has a
surjective coordinatewise interpretation on $A^m$.
\end{corollary}

\begin{proof}
Iterate Lemma~\ref{lem:product-surjectivity}.
\end{proof}

\begin{theorem}[Structure of the perfect-alphabet spectrum]
\label{thm:perfect-spectrum-structure}
Let $\mathbf t$ have $r$ output terms.
\begin{enumerate}[label=\textit{(\roman*)}, leftmargin=*]
\item $\PerfSpec(\mathbf t)$ is multiplicatively closed; in particular, if $n\in\PerfSpec(\mathbf t)$, then
      $n^m\in\PerfSpec(\mathbf t)$ for every $m\ge1$.
\item If $\PerfSpec(\mathbf t)$ is nonempty, then $D(\mathbf t)=r$.
\item The converse to \textit{(ii)} fails: full asymptotic rate does not imply a nonempty perfect spectrum.
\end{enumerate}
\end{theorem}

\begin{proof}
For \textit{(i)}, if $a,b\in\PerfSpec(\mathbf t)$, choose surjective interpretations on alphabets $A,B$ of
those sizes and apply Lemma~\ref{lem:product-surjectivity}; the product alphabet has size $ab$.  The statement
about powers follows by induction.

For \textit{(ii)}, take $n\in\PerfSpec(\mathbf t)$.  By \textit{(i)},
$n^m\in\PerfSpec(\mathbf t)$ for every $m\ge1$, and therefore
\[
\frac{\log\Disp_{n^m}(\mathbf t)}{\log(n^m)}=r.
\]
The limit defining $D(\mathbf t)$ exists by the dispersion embedding
(Lemma~\ref{lem:disp-embed}) and the term-coding convergence theorem
\parencite{riis2026termcoding}.  Since $n^m\to\infty$, the full limit must equal $r$.

For \textit{(iii)}, the diamond tuple has $D(\mathbf t)=4$ by Lemma~\ref{lem:diamond-maxflow}, whereas
Lemma~\ref{lem:diamond-not-bij} shows that its perfect spectrum is empty.
\end{proof}

\begin{corollary}[Perfect dispersion implies full rate]
\label{thm:perfect-implies-full-rate}
\[
\textsc{Perfect-Disp}\ \subseteq\ \textsc{Full-Rate-Disp},
\]
and the inclusion is strict.
\end{corollary}

\begin{proof}
This is Theorem~\ref{thm:perfect-spectrum-structure}\textit{(ii)--(iii)}.
\end{proof}

\begin{corollary}[Full-rate dispersion is in polynomial time]
\label{cor:full-rate-ptime}
The problem $\textsc{Full-Rate-Disp}$ is decidable in polynomial time.
\end{corollary}

\begin{proof}
Compute the integer $D(\mathbf t)=\rho(\mathbf t)$ by the vertex-split max-flow algorithm of
Theorem~\ref{thm:disp-exponent} and compare it with the number $r$ of output terms.
\end{proof}

\begin{theorem}[Complete one-output classification]
\label{thm:one-output}
Let $\mathbf t=(t_1)$ have one output.  If some input variable occurs in $t_1$, then
\[
\PerfSpec(\mathbf t)=\{2,3,4,\ldots\};
\]
if no input variable occurs, then $\PerfSpec(\mathbf t)=\varnothing$.
Consequently, $\textsc{Perfect-Disp}_1$ is decidable in linear time by scanning the term.
\end{theorem}

\begin{proof}
If $t_1$ contains no variable, it evaluates to a constant under every interpretation, so its image has size one.
Conversely, suppose that $x_j$ occurs in $t_1$.  Fix any finite alphabet $A$ with $|A|\ge2$, choose $c\in A$,
and restrict to the input slice on which every variable except $x_j$ equals $c$.  For each occurrence
$u=f(u_1,\ldots,u_q)$, let $P$ be the child positions whose subterm contains $x_j$, and prescribe
\[
f^\mathcal I(\pi_P(a))=
\begin{cases}
a,&P\ne\varnothing,\\
c,&P=\varnothing,
\end{cases}
\qquad
\pi_P(a)_i=
\begin{cases}
a,&i\in P,\\
c,&i\notin P.
\end{cases}
\]
These prescriptions remain consistent when the same symbol occurs several times.  Indeed, if
$\pi_P(a)=\pi_Q(b)$, then either the non-$c$ coordinates determine $P=Q$ and $a=b$, or both sides are the
all-$c$ tuple and both prescribed outputs are $c$.  Extend each operation arbitrarily away from the prescribed
tuples.  Induction on subterms now gives $t_1^\mathcal I=x_j$ on the chosen slice, so $t_1^\mathcal I$ is
surjective.  The syntactic occurrence test is linear in the term size.
\end{proof}

\subsection{Padding lemmas}
\label{subsec:padding-lemmas}

\begin{lemma}[Padding outputs preserves surjectivity equivalence]
\label{lem:padding-outputs}
Let $A$ be a nonempty finite set, let $p: A^{k_0} \to A^{k_0}$ be a map, and let $k \ge r \ge k_0$.
Define $F: A^k \to A^r$ by
\[
F(a_1, \dots, a_k) := \bigl(p(a_1, \dots, a_{k_0}),\, a_{k_0+1}, \dots, a_r\bigr).
\]
Then $F$ is surjective if and only if $p$ is surjective.
\end{lemma}

\begin{proof}
($\Rightarrow$) If $F$ is surjective, then for any $(b_1,\dots,b_{k_0})\in A^{k_0}$ we may choose arbitrary trailing
coordinates $b_{k_0+1},\dots,b_r$ and use surjectivity of $F$ to find
\[
(a_1,\dots,a_k)\in A^k
\]
with
\[
F(a_1,\dots,a_k)=(b_1,\dots,b_r).
\]
By the definition of $F$, this implies
\[
p(a_1,\dots,a_{k_0})=(b_1,\dots,b_{k_0}),
\]
so $p$ is surjective.

($\Leftarrow$) If $p$ is surjective, then for any $(b_1,\dots,b_r)\in A^r$ there exist
\[
(a_1,\dots,a_{k_0})\in A^{k_0}
\]
such that
\[
p(a_1,\dots,a_{k_0})=(b_1,\dots,b_{k_0}).
\]
Taking
\[
a_{k_0+1}=b_{k_0+1},\dots,a_r=b_r,
\]
and choosing the remaining inputs arbitrarily if $k>r$, we obtain
\[
F(a_1,\dots,a_k)=(b_1,\dots,b_r).
\]
Thus $F$ is surjective.
\end{proof}

\begin{lemma}[Unused inputs are irrelevant for surjectivity]
\label{lem:unused-inputs}
Let $A$ be a nonempty finite set and let $F: A^k \to A^r$ factor through the projection to the first
$r$ coordinates, i.e.
\[
F(a_1,\dots,a_k)=f(a_1,\dots,a_r)
\]
for some $f:A^r\to A^r$.
Then $F$ is surjective if and only if $f$ is surjective.
\end{lemma}

\begin{proof}
Immediate from
\[
\Img(F)=\Img(f).
\]
\end{proof}

\subsection{Reductions from square bijectivity to perfect dispersion}
\label{subsec:reductions}

\begin{theorem}[Reduction in the square case]
\label{thm:reduction-kk}
For every fixed $k\ge 3$, there is a polynomial-time many-one reduction
\[
\textsc{BIJ}_3 \le_m \textsc{Perfect-Disp}_k
\]
on square instances (instances with exactly $k$ input variables and $k$ output terms).
In particular, if $\textsc{BIJ}_3$ is undecidable then $\textsc{Perfect-Disp}_k$ is undecidable.
\end{theorem}

\begin{proof}
Fix $k\ge 3$.
We reduce from $\textsc{BIJ}_3$.
Given an instance
\[
\mathbf p(x_1,x_2,x_3)=(p_1(x_1,x_2,x_3),p_2(x_1,x_2,x_3),p_3(x_1,x_2,x_3)),
\]
form the square $k$-tuple
\[
\mathbf t(x_1,\dots,x_k):=
\bigl(
p_1(x_1,x_2,x_3),\,
p_2(x_1,x_2,x_3),\,
p_3(x_1,x_2,x_3),\,
x_4,\dots,x_k
\bigr).
\]
The transformation introduces no new function symbols and has
$\instsize{\mathbf t}=\instsize{\mathbf p}+2(k-3)$, so it is computable in linear time and has linear output
size.  For every interpretation $\mathcal I$ on a finite set $A$, use the same $\mathcal I$ for the original
and transformed instances.  The induced map
\[
\Theta^\mathcal I:A^k\to A^k
\]
is of the form covered by Lemma~\ref{lem:padding-outputs}, with $k_0=3$ and $r=k$.
Hence $\Theta^\mathcal I$ is surjective if and only if the original map
\[
A^3\to A^3,\qquad
(a_1,a_2,a_3)\mapsto
\bigl(
p_1^\mathcal I(a_1,a_2,a_3),\,
p_2^\mathcal I(a_1,a_2,a_3),\,
p_3^\mathcal I(a_1,a_2,a_3)
\bigr)
\]
is surjective.
Since $A^3$ is finite, that is equivalent to bijectivity.
By Lemma~\ref{lem:square-perfect-equals-bij}, the constructed instance has perfect dispersion on some finite
alphabet if and only if the original $\textsc{BIJ}_3$ instance is a yes-instance.
Therefore
\[
\textsc{BIJ}_3 \le_m \textsc{Perfect-Disp}_k
\]
on square instances.
This establishes the many-one reduction $\textsc{BIJ}_3 \le_m \textsc{Perfect-Disp}_k$.
\end{proof}

\begin{remark}[The case $k<r$]
\label{rem:k-less-r}
If a dispersion instance has $k<r$, then for every alphabet size $n$ and every interpretation $\mathcal I$ we have
\[
|\Img(\Theta^\mathcal I)|\le |A^k|=n^k<n^r,
\]
so perfect dispersion is impossible.
Therefore all exact perfect-dispersion questions may be restricted to the case $k\ge r$.
\end{remark}

\begin{theorem}[Reduction for fixed output dimension]
\label{thm:reduction-nr}
For every fixed pair $k\ge r\ge 3$, there is a polynomial-time many-one reduction
\[
\textsc{BIJ}_3 \le_m \textsc{Perfect-Disp}_r
\]
on the subfamily of instances with exactly $k$ input variables and $r$ output terms.
Consequently, if $\textsc{BIJ}_3$ is undecidable then $\textsc{Perfect-Disp}_r$ is undecidable.
\end{theorem}

\begin{proof}
Fix integers $k\ge r\ge 3$.
We again reduce from $\textsc{BIJ}_3$.
Given an instance
\[
\mathbf p(x_1,x_2,x_3)=(p_1(x_1,x_2,x_3),p_2(x_1,x_2,x_3),p_3(x_1,x_2,x_3)),
\]
construct a dispersion instance with $k$ inputs and $r$ outputs by
\[
t_i(\mathbf x)=
\begin{cases}
p_i(x_1,x_2,x_3) & \text{for } i=1,2,3,\\
x_i & \text{for } 4\le i\le r.
\end{cases}
\]
The remaining input variables $x_{r+1},\dots,x_k$ do not occur in the outputs; by
Lemma~\ref{lem:unused-inputs}, they do not affect surjectivity.  The construction introduces no new function
symbols and satisfies
$\instsize{\mathbf t}=\instsize{\mathbf p}+2(r-3)$, so it is computable in linear time and has linear output
size.

For every interpretation $\mathcal I$ on a finite set $A$, use the same $\mathcal I$ for both instances.  The induced map
\[
\Theta^\mathcal I:A^k\to A^r
\]
has the form treated by Lemma~\ref{lem:padding-outputs}, again with $k_0=3$.
Therefore $\Theta^\mathcal I$ is surjective if and only if the original map
\[
A^3\to A^3,\qquad
(a_1,a_2,a_3)\mapsto
\bigl(
p_1^\mathcal I(a_1,a_2,a_3),\,
p_2^\mathcal I(a_1,a_2,a_3),\,
p_3^\mathcal I(a_1,a_2,a_3)
\bigr)
\]
is surjective, hence bijective.
Thus the constructed instance has perfect dispersion on some finite alphabet if and only if the original
$\textsc{BIJ}_3$ instance is a yes-instance.
This gives a many-one reduction
\[
\textsc{BIJ}_3 \le_m \textsc{Perfect-Disp}_r
\]
on the subfamily with exactly $k$ inputs.
This establishes the many-one reduction $\textsc{BIJ}_3 \le_m \textsc{Perfect-Disp}_r$.
\end{proof}

\begin{remark}[The borderline output dimension $r=2$]
Theorem~\ref{thm:one-output} resolves $r=1$, and the padding reductions above cover every fixed $r\ge3$.
For $r=2$, Corollary~\ref{cor:bij-reduces-perfect} identifies the square subfamily of
$\textsc{Perfect-Disp}_2$ exactly with $\textsc{BIJ}_2$.  This observation does not yield a classification,
since $\textsc{BIJ}_2$ itself remains open, and the padding reductions from $\textsc{BIJ}_3$ cannot decrease
the output dimension from three to two.  Thus the decidability of $\textsc{Perfect-Disp}_2$ remains open.
\end{remark}

\subsection{The unconditional asymptotic theorem}
\label{subsec:unconditional-asymptotic}

Fix an integer $d \ge 0$ and a threshold function $h_d: \mathbb{N} \to \mathbb{N}$ satisfying:
\begin{equation}
\label{eq:threshold-conditions}
\exists N_0\ \forall n \ge N_0:\ n^d < h_d(n), \qquad \text{and} \qquad h_d(n) \in o(n^{d+1}).
\end{equation}
For example, $h_d(n) = n^d + 1$ satisfies \eqref{eq:threshold-conditions}; so does
$h_d(n)=\lfloor n^d \log n \rfloor$, where the logarithm may have any fixed base greater than one.

\begin{theorem}[Polynomial-time decidability of asymptotic thresholds]
\label{thm:ptime-threshold}
Fix $d \ge 0$ and a threshold function $h_d$ satisfying \eqref{eq:threshold-conditions}. Then the problem:
\begin{quote}
\textbf{Input:} A dispersion instance $\mathbf t$.\\
\textbf{Question:} Does there exist $N$ such that for all $n \ge N$: $\Disp_n(\mathbf t) \ge h_d(n)$?
\end{quote}
is decidable in polynomial time in the size of $\mathbf t$.
\end{theorem}

\begin{proof}
By Theorem~\ref{thm:disp-exponent}, we have
\[
\Disp_n(\mathbf t)=\Theta(n^{D(\mathbf t)})
\]
where $D(\mathbf t)\in\mathbb Z_{\ge 0}$ is computable in polynomial time.

We claim that the answer is ``yes'' if and only if
\[
D(\mathbf t)\ge d+1.
\]

Assume first that $D(\mathbf t)\ge d+1$.
Then there exist constants $c>0$ and $N_1$ such that for all $n\ge N_1$,
\[
\Disp_n(\mathbf t)\ge c\,n^{D(\mathbf t)}\ge c\,n^{d+1}.
\]
Since $h_d(n)\in o(n^{d+1})$, there exists $N_2$ such that for all $n\ge N_2$,
\[
h_d(n)\le c\,n^{d+1}.
\]
Hence for all $n\ge \max\{N_1,N_2\}$ we have
\[
\Disp_n(\mathbf t)\ge h_d(n).
\]

Conversely, assume that $D(\mathbf t)\le d$.
The upper-bound part of Theorem~\ref{thm:disp-exponent} gives
\[
\Disp_n(\mathbf t)\le n^{D(\mathbf t)}\le n^d
\]
for all $n$.
By \eqref{eq:threshold-conditions}, there exists $N_0$ such that for all $n\ge N_0$,
\[
n^d<h_d(n).
\]
Therefore for all $n\ge N_0$,
\[
\Disp_n(\mathbf t)\le n^d<h_d(n),
\]
so the asymptotic threshold condition fails.
Notice that this direction uses the cut bound with constant exactly one.  A generic upper estimate
$\Disp_n(\mathbf t)\le Cn^d$ with $C>1$ would not suffice for the threshold $n^d+1$.

Thus the decision procedure is:
\begin{enumerate}[leftmargin=*]
\item compute $D(\mathbf t)$ using the vertex-split max-flow algorithm of Section~\ref{subsec:compute-exponent};
\item return ``yes'' iff $D(\mathbf t)\ge d+1$.
\end{enumerate}
This runs in polynomial time.
\end{proof}

\begin{remark}
Theorem~\ref{thm:ptime-threshold} is unconditional.
The key difference from the exact side is that it asks about eventual asymptotic growth rather than the existence of
a finite alphabet size at which full codomain volume is attained.
\end{remark}

\subsection{Summary of the complexity landscape}
\label{subsec:summary-landscape}

\begin{theorem}[Exact-vs-rate complexity landscape]
\label{thm:landscape}
The following hold for single-sorted dispersion:
\begin{enumerate}[label=\textit{(\roman*)}, leftmargin=*]
\item the perfect-alphabet spectrum is multiplicatively closed; a nonempty spectrum implies full rate, and the
      converse fails;
\item $\textsc{Perfect-Disp}_1$ is decidable in linear time, with spectrum either empty or all sizes $n\ge2$;
\item the square restriction of $\textsc{Perfect-Disp}$ is exactly $\textsc{BIJ}$;
\item for every fixed $r\ge 3$, there is a polynomial-time many-one reduction
\[
\textsc{BIJ}_3 \le_m \textsc{Perfect-Disp}_r.
\]
Consequently, undecidability of $\textsc{BIJ}_3$ would imply undecidability of
$\textsc{Perfect-Disp}_r$;
\item $\textsc{Full-Rate-Disp}$ is decidable in polynomial time, and for every fixed $d\ge 0$ and every
threshold function $h_d$ satisfying
\eqref{eq:threshold-conditions}, the asymptotic threshold problem
$\textsc{Asymp-Threshold}_{d,h_d}$ is decidable in polynomial time.
\end{enumerate}
\end{theorem}

\begin{proof}
Part \textit{(i)} is Theorem~\ref{thm:perfect-spectrum-structure}; part \textit{(ii)} is
Theorem~\ref{thm:one-output}; part \textit{(iii)} is Corollary~\ref{cor:bij-reduces-perfect}; part
\textit{(iv)} follows from Theorems~\ref{thm:reduction-kk} and~\ref{thm:reduction-nr}; and part \textit{(v)}
is Corollary~\ref{cor:full-rate-ptime} together with Theorem~\ref{thm:ptime-threshold}.
\end{proof}

\section{Square bijectivity: tractable fragments and the open frontier}
\label{sec:bij-status}

This section proves two positive fragments of the square problem before isolating the unresolved frontier.  The
one-dimensional problem is completely classified, explicit two-dimensional tuples exhibit nontrivial perfect
spectra, and the scalar-linear fragment has an exact determinant criterion.  None of these results supplies a
decision procedure for unrestricted $\textsc{BIJ}$ or for $\textsc{BIJ}_k$ with fixed $k\ge2$.

\subsection{Low-dimensional spectra}

\begin{corollary}[The one-dimensional square problem]
\label{cor:bij-one}
A single term $t(x)$ induces a bijection under some finite interpretation if and only if $x$ occurs in $t$.
If $x$ occurs, every alphabet size $n\ge2$ admits a witness; otherwise none does.
\end{corollary}

\begin{proof}
This is the square specialization of Theorem~\ref{thm:one-output}.
\end{proof}

For $k=2$, the general decision problem is open, but exact spectra are already nontrivial.  Consider
\[
\boldsymbol\tau(x,y)=\bigl(f(f(x,y),f(x,x)),f(x,x)\bigr),
\]
and
\[
\boldsymbol\sigma(x,y)=\bigl(f(x,y),f(f(x,x),f(y,y))\bigr),
\]

\begin{proposition}[Two exact least-alphabet results]
\label{prop:small-spectra}
The least element of $\PerfSpec(\boldsymbol\tau)$ is $3$, and the least element of
$\PerfSpec(\boldsymbol\sigma)$ is $4$.  Consequently,
\[
\{3^m:m\ge1\}\subseteq\PerfSpec(\boldsymbol\tau),
\qquad
\{4^m:m\ge1\}\subseteq\PerfSpec(\boldsymbol\sigma).
\]
\end{proposition}

\begin{proof}
For each alphabet size $n=2,3$, a complete check consists of enumerating the $n^{n^2}$ binary operation tables
$f:A^2\to A$ and, for every table, testing the $n^2$ input pairs for distinct outputs.  This is only $16$
tables for $n=2$ and $19{,}683$ tables for $n=3$.  The exhaustive checks give
$2\notin\PerfSpec(\boldsymbol\tau)$ and
$2,3\notin\PerfSpec(\boldsymbol\sigma)$.

For $\boldsymbol\tau$ on $\mathbb F_3$, take $f(x,y)=x+y$.  Then
\[
\boldsymbol\tau(x,y)=(y,2x),
\]
which is bijective, so $3\in\PerfSpec(\boldsymbol\tau)$.  For $\boldsymbol\sigma$, a witness on
$\{0,1,2,3\}$ is the following table, with rows indexed by the first argument and columns by the second:
\[
\begin{array}{c|cccc}
f & 0&1&2&3\\ \hline
0 & 2&2&3&2\\
1 & 1&0&1&3\\
2 & 0&1&1&2\\
3 & 3&0&0&3
\end{array}
\]
Direct substitution verifies that the resulting map $A^2\to A^2$ is bijective.  Thus
$4\in\PerfSpec(\boldsymbol\sigma)$.  The assertions about powers follow from
Theorem~\ref{thm:perfect-spectrum-structure}\textit{(i)}.
\end{proof}

\subsection{A complete scalar-linear classification}
\label{subsec:linear-decidable}
Interpret each $r$-ary function symbol by
\[
s(v_1,\dots,v_r)=\lambda_1v_1+\cdots+\lambda_rv_r
\]
with scalar coefficients in a field $K$.
Then every term becomes a linear form in the input variables.
Introducing commuting indeterminates $\lambda_{s,j}$ for the argument-slots of the symbols in the signature, one
defines recursively a symbolic coefficient matrix
\[
C_{\mathbf t}\in M_k\bigl(\mathbb Z[\lambda_{s,j}]\bigr)
\]
for every square tuple $\mathbf t=(t_1,\dots,t_k)$, and writes
\[
\Delta_{\mathbf t}:=\det C_{\mathbf t}.
\]

\begin{theorem}[Scalar-linear determinant criterion]
\label{prop:scalar-det-criterion-main}
For every field $K$, the tuple $\mathbf t$ has a scalar-linear witness over $K$ if and only if there exists an
assignment of the scalar coefficients in $K$ for which
\[
\Delta_{\mathbf t}(a)\neq 0\quad\text{in }K.
\]
\end{theorem}

\begin{proof}
Under a scalar-linear interpretation, each term $t_i(\mathbf x)$ expands by induction to a linear form
\[
t_i(\mathbf x)=\sum_{j=1}^k c_{ij}\,x_j,
\]
where the coefficients $c_{ij}$ are obtained from the recursive construction of $C_{\mathbf t}$.
Hence the induced map on $K^k$ is represented by the matrix $C_{\mathbf t}(a)$ obtained by evaluating the symbolic
coefficients at the chosen scalars $a=(\lambda_{s,j})$.
Bijectivity is therefore equivalent to invertibility of this matrix, i.e. to
\[
\det C_{\mathbf t}(a)\neq 0.
\]
\end{proof}

\begin{corollary}[Finite-field classification of scalar-linear witnesses]
\label{cor:scalar-linear-classification}
The following hold.
\begin{enumerate}[label=\textit{(\alph*)}, leftmargin=*]
\item Over any fixed finite field $K$, scalar-linear bijectivity is decidable.
\item For a fixed prime $p$, a scalar-linear witness exists over some finite extension of $\mathbb F_p$ if and
      only if $\Delta_{\mathbf t}$ is nonzero modulo $p$.
\item A scalar-linear witness exists over some finite field, of arbitrary characteristic, if and only if
      $\Delta_{\mathbf t}$ is a nonzero polynomial in $\mathbb Z[\lambda_{s,j}]$.
\end{enumerate}
\end{corollary}

\begin{proof}
Part \textit{(a)} is finite search over the coefficient assignments.  For \textit{(b)}, a nonzero polynomial
over $\mathbb F_p$ has a nonzero evaluation over a sufficiently large finite extension: choose an extension
whose cardinality exceeds the degree in every variable.  Conversely, an identically zero reduction cannot
have a nonzero evaluation.  For \textit{(c)}, if the integer polynomial is nonzero, choose a prime $p$ that
does not divide one of its nonzero coefficients and apply \textit{(b)}; the converse is immediate.
\end{proof}

\subsection{The matrix-linear frontier}
\label{subsec:matrix-linear}
Restrict attention to \emph{matrix-linear} witnesses: interpret each $r$-ary function symbol by a linear map
$s(v_1,\dots,v_r)=A_1v_1+\cdots+A_rv_r$ where each $A_i\in M_d(K)$ for some field $K$ and some $d\ge 1$.
The induced square map $\mathbf t^\A$ is then a linear map on $(K^d)^k$, and bijectivity is equivalent to
invertibility of an $dk\times dk$ block matrix whose entries are noncommutative polynomials in the coefficient
matrices.

Higman linearisation suggests converting this block matrix to an affine linear pencil and then testing fullness
over a free skew field.  Deterministic polynomial-time algorithms for the resulting nc-rank problem are known
\parencite{garg2016deterministic,ivanyos2018constructive}.  A complete application here would still have to
spell out the linearisation and the passages between base fields, extensions, and characteristics.  We do not
undertake those steps in this paper, and make no theorem-level claim about the matrix-linear decision problem or
its characteristic spectrum.  This perspective is included only as a possible route for future work.

\subsection{Why undecidability is not straightforward}
\label{subsec:why-hard}
The most natural approach to proving undecidability of $\textsc{BIJ}_k$ would be to reduce from a known
undecidable problem by encoding equality constraints as bijectivity conditions.  However, bijectivity is naturally
good at expressing \emph{non-collapse} (injectivity, recoverability, orthogonality) but not at forcing pointwise
equalities.  This ``polarity mismatch'' is a persistent obstacle: a local full-image constraint that is supposed
to force $u=v$ can typically be defeated by duplicating a successful row at a different element $v\neq u$.

The closest conceptual match is the undecidability of finite developability for permutoids \parencite{bridson2017undecidability}, which
asks whether a finite system of partial bijections can be completed to honest permutations on a finite set.
However, even a successful encoding via a Cameron-atlas reduction would, in its naive form, be compatible with
field-based matrix-linear witnesses (since finite permutations lift to permutation matrices).  Any undecidability
proof would therefore need to combine an atlas encoding with a nonlinear obstruction that blocks the linear witnesses.

\begin{remark}[The proof sketch in the earlier version]
\label{rem:earlier-proof-sketch}
We have not been able to formalise the counter-machine proof sketch from the first arXiv version into a complete
rigorous proof.  That argument was explicitly headed \enquote{Detailed proof sketch (structure suitable for
formalisation)}.  It described the intended simulation, but did not prove the forcing step needed to show that
an arbitrary bijective interpretation of the auxiliary symbols follows the intended machine dynamics.
Accordingly, we withdraw the earlier claim that the finite square-bijectivity decision problem is
undecidable, together with the unconditional undecidability claims for perfect dispersion that depended on it.
No result in the present paper uses the sketch.  The decidability of $\textsc{BIJ}$ and of $\textsc{BIJ}_k$ for
fixed $k\ge2$ remains open.
\end{remark}

\subsection{Current assessment}
The problems $\textsc{BIJ}$ and $\textsc{BIJ}_k$ for fixed $k\ge 2$ remain open.  The decidable linear
fragments, the structural constraints on bijectivity spectra (multiplicative closure, universal-sentence origin,
and a purely functional language), and the polarity mismatch all obstruct the most direct reduction strategies.
None of these facts proves decidability.  Conversely, no sound finite-model compiler establishing
undecidability is currently available.  The correct conclusion is therefore an open-problem statement, not a
directional claim.

\section*{Declaration of generative AI and AI-assisted technologies in the manuscript preparation process}
During the preparation of this work, the author used Anthropic Claude and OpenAI Codex to assist with manuscript
organisation, language editing, critical auditing, and the preparation and review of the Lean formalisation.
The author reviewed and edited all resulting content as needed and takes full responsibility for the content of
the publication.  The machine-verification statements in Section~\ref{subsec:verification} refer to checking by
the Lean kernel.

\bibliographystyle{plainnat}
\bibliography{termcoding_submission}

\end{document}